\documentclass[11pt]{article}
\usepackage[english]{babel}
\usepackage{amsmath}
\usepackage{amssymb}
\usepackage{amsthm}
\usepackage{authblk}
\usepackage{amsfonts}
\usepackage{fullpage}
\usepackage{color}

\newcommand{\vertiii}[1]{{\left\vert\kern-0.25ex\left\vert\kern-0.25ex\left\vert #1 
    \right\vert\kern-0.25ex\right\vert\kern-0.25ex\right\vert}}

\newcommand{\ket}[1]{|#1\rangle}

\usepackage{cite}
\usepackage{comment}
\usepackage[colorlinks=true]{hyperref}
\newtheorem{thm}{Theorem}

\newtheorem*{thm*}{Theorem}
\newtheorem{cor}{Corollary}

\newtheorem{lem}{Lemma}

\newtheorem{prop}{Proposition}

\theoremstyle{definition}
\newtheorem{defn}{Definition}

\theoremstyle{remark}
\newtheorem{rem}{Remark}

\newcommand{\tr}{\mathrm{Tr}}
\newcommand{\cN}{\mathcal{N}}

\newcommand{\cB}{\mathcal{B}}

\newcommand{\cP}{\mathcal{P}}

\newcommand{\cH}{\mathcal{H}}

\newcommand{\id}{\mathrm{id}}
\usepackage{graphicx}
\graphicspath{ {./images/} }

\begin{document}
\title{Quantum concentration inequalities}
\author[1]{Giacomo De Palma\thanks{giacomo.depalma@sns.it}}
\author[2]{Cambyse Rouz\'e\thanks{cambyse.rouze@tum.de}}
\affil[1]{Scuola Normale Superiore, 56126 Pisa, Italy }
\affil[2]{Technische Universit\"at M\"unchen, 85748 Garching, Germany}

\maketitle
\begin{abstract}
We establish Transportation Cost Inequalities (TCIs) with respect to the quantum Wasserstein distance by introducing quantum extensions of well-known classical methods: First, we generalize the Dobrushin uniqueness condition to prove that Gibbs states of 1D commuting Hamiltonians satisfy a TCI at any positive temperature and provide conditions under which this first result can be extended to non-commuting Hamiltonians. Next, using a non-commutative version of Ollivier's coarse Ricci curvature, we prove that high temperature Gibbs states of commuting Hamiltonians on arbitrary hypergraphs $H=(V,E)$ satisfy a TCI with constant scaling as $O(|V|)$. Third, we argue that the temperature range for which the TCI holds can be enlarged by relating it to recently established modified logarithmic Sobolev inequalities. Fourth, we prove that the inequality still holds for fixed points of arbitrary reversible local quantum Markov semigroups on regular lattices, albeit with slightly worsened constants, under a seemingly weaker condition of local indistinguishability of the fixed points. Finally, we use our framework to prove Gaussian concentration bounds for the distribution of eigenvalues of quasi-local observables and argue the usefulness of the TCI in proving the equivalence of the canonical and microcanonical ensembles and an exponential improvement over the weak Eigenstate Thermalization Hypothesis.
\end{abstract}

\section{Introduction}

Given a random variable $X$ of law $\mu$ taking values on a metric space $(\Omega,d_\Omega)$ and a function $f:\Omega\to \mathbb{R}$, a concentration of measure inequality quantifies the probability that the random variable $f(X)$ deviates from its mean or its median. Since the early age of the theory, concentration inequalities have seen many new methods, refinements and exciting applications to various areas of mathematics
\cite{raginsky2012concentration,boucheron2013concentration,kontorovich2017concentration}. Among the different classes of concentration inequalities, Gaussian concentration is arguably the most standard one: the measure $\mu$ is said to be 
\textit{sub-Gaussian} if there exist constants $K,\kappa>0$ such that, for all $A\subseteq \Omega$ with $\mu(A)\ge1/2$ we have for any $r\ge0$
\begin{align}\label{eq:Gaussconc}
    \mu\left(\{x\in \Omega:\,d_\Omega(x,A)<r\} \right)\ge 1-Ke^{-\kappa r^2}\,.
\end{align}
In her seminal work \cite{marton1986simple}, Marton made the beautiful observation that the above behavior can be obtained as a consequence of a \textit{transportation cost} inequality: if there exists $c>0$ such that, for any probability measure $\nu<\!<\mu$, 
\begin{align}\label{eq:wassers}\tag{$\operatorname{TC}(c)$}
    W_1(\mu,\nu)\le \sqrt{c \,S(\nu\|\mu)}\,,
\end{align}
then \eqref{eq:Gaussconc} holds with constants $\kappa=\frac{1}{c}$ and $K=1$ for all $r>\sqrt{c\ln 2}$. Here, $S(\nu\|\mu)$ refers to the relative entropy between the measures $\nu$ and $\mu$, whereas the quantity $W_1(\mu,\nu)$ in \eqref{eq:wassers} is the Wasserstein distance between the two measures $\mu,\nu$:
\begin{align*}
    W_1(\mu,\nu):=\sup_{\|f\|_L\le 1}\,\big|\mathbb{E}_\mu[f]-\mathbb{E}_\nu[f]\big|\,.
\end{align*}
Later, \cite{bobkov1999exponential} proved that transportation cost inequalities are in fact equivalent to the property of sub-Gaussianity: more precisely, \eqref{eq:wassers} holds if and only if for all Lipschitz functions $f:\Omega\to\mathbb{R}$,
\begin{align*}
    \mathbb{P}_\mu\Big(\big|f(X)-\mathbb{E}_\mu[f(X)]\big|>t  \Big)\le 2\,e^{-\frac{t^2}{c}}\,,\forall t\ge 0\,.
\end{align*}
One of the main advantages of transportation cost inequalities is their tensorization property: assume that $\mu$ satisfies $\operatorname{TC}(c)$, then $\mu^{\otimes n}$ satisfies $\operatorname{TC}(nc)$ for all $n\in\mathbb{N}$, where the set $\Omega^n$ is provided with the metric 
\begin{align*}
    d_n(x^n,y^n):=\sum_{i=1}^n\,d_{\Omega}(x_i,y_i)\,.
\end{align*}
Perhaps the simplest example of that sort is given by taking $\Omega^n=[d]^n$ endowed with the Hamming distance $d_H$. In the case $n=1$, the corresponding Wasserstein distance reduces to the total variation, and $\operatorname{TC}(1/2)$ holds for any measure $\mu$, since it simply reduces to Pinsker's inequality. For $n\ge 1$, $\mu^{\otimes n}$ satisfies $\operatorname{TC}(n/2)$.

While the theory of concentration inequalities for i.i.d.~random variables is by now well understood, things become more challenging when the random variables are allowed to depend on each other \cite{kontorovich2017concentration,dobrushin1970prescribing}. One way to extend concentration bounds to weakly dependent random variables is to assume that their joint law $\mu$ satisfies the so-called Dobrushin uniqueness condition \cite{dobrushin1970prescribing}. Dobrushin's uniqueness condition plays an important role in the study of Gibbs measures in the one-phase region, however it often turns out to be a very strong requirement on the measure $\mu$. More recently, Marton gave an attempt at extending the i.i.d.~theory beyond the mere Gibbs setting \cite{marton2019logarithmic}. Her main result consists in a logarithmic Sobolev inequality for a generic measure $\mu$ - well known to imply transportation cost inequalities - under the so-called Dobrushin--Shlosman mixing condition \cite{dobrushin1987completely}, the latter condition being weaker than Dobrushin's uniqueness condition. As mentioned in \cite{kulske2003concentration}, such paths to establish Gaussian concentration suffer from the difficulty of deriving explicit constants. Moreover, the result of \cite{marton2019logarithmic} also relies on the crucial assumption that the measure $\mu$ has full support.

Recently, concentration inequalities have attracted much attention in the communities of random matrix theory, quantum information theory and operator algebras \cite{Junge2014,Tropp2015,Huang2021,Rouz2019,de2020quantum,DePalma2021,Golse2016,cole2021quantum,carlen2017gradient}. In \cite{de2020quantum}, a quantum Wasserstein distance of order 1  (or quantum $W_1$ distance) was defined on the set of the quantum states of $n$ qudits with the property that it strictly reduces to the classical Wasserstein distance on $[d]^n$ for states that are diagonal in the computational basis.
This quantum generalization of the Wasserstein distance is based on the notion of neighboring states.
Two quantum states of $n$ qudits are neighboring if they differ only in one qudit, \emph{i.e.}, if they coincide after that qudit is discarded.
The quantum $W_1$ distance is then that induced by the maximum norm that assigns distance at most one to every couple of neighboring states \cite[Definition 4]{de2020quantum}.
Such norm is called quantum $W_1$ norm and is denoted with $\left\|\cdot\right\|_{W_1}$.
The quantum $W_1$ norm proposed in Ref. \cite{de2020quantum} admits a dual formulation in terms of a quantum generalization of the Lipschitz constant.
Denoting with $\mathcal{O}_n$ the set of the observables of $n$ qudits, the Lipschitz constant of the observable $H\in\mathcal{O}_n$ is defined as \cite[Section V]{de2020quantum} 
\begin{align}
    \|H\|_L:=2\max_{i\in [n]}\min\left\{\left\|H- H^{(i)}\right\|_\infty:H^{(i)}\in\mathcal{O}_n\text{ does not act on the $i$-th qudit}\right\}\,.
\end{align}
Then, the quantum $W_1$ distance between the states $\rho$ and $\omega$ can also be expressed as \cite[Section V]{de2020quantum}
\begin{align}
  \left\|\rho-\omega\right\|_{W_1}  =\max\left\{\tr\left[\left(\rho-\omega\right)H\right]:\,H\in\mathcal{O}_n,\,\|H\|_L\le 1\right\}\,.
\end{align}
Moreover, in \cite{de2020quantum} it was showed that $\operatorname{TC}(n/2)$ holds for any tensor product $\omega=\omega_1\otimes\ldots \otimes \omega_n$ of quantum states, hence extending Marton's original inequality with the exact same constant: for any state $\rho$ of $n$ qudits, 
\begin{align*}
    \left\|\rho-\omega_1\otimes\ldots\otimes \omega_n\right\|_{W_1}\le \sqrt{\frac{n}{2}\,S(\rho\|\omega_1\otimes\ldots\otimes \omega_n)}\,,
\end{align*}
where $S(\rho\|\omega):=\tr[\rho\,(\ln\rho-\ln\omega)]$ denotes Umegaki's relative entropy between the states $\rho$ and $\omega$. 

\paragraph{Main results:}
In this paper, we prove that any of the following conditions implies a transportation cost inequality:
\begin{itemize}
    \item[(i)] A non-commutative Dobrushin uniqueness condition (\autoref{transportationcostdobrushin});
    \item[(ii)] A generalization of Ollivier's coarse Ricci curvature bound (\autoref{Lipcontract});
    \item[(iii)] A modified logarithmic Sobolev inequality condition (\autoref{MLSI});
    \item[(iv)] A condition of local indistinguishability of the state (\autoref{localindistinguishability}).
\end{itemize}

Each of these methods comes with its strengths and weaknesses:
\begin{itemize}
    \item[(i)] The non-commutative Dobrushin uniqueness condition implies a nontrivial TCI at any temperature (see \autoref{rem:eta}), but the scaling of the constant $c$ with the number of subsystems is optimal only in one dimension (see \autoref{rem:1D}).
    \item[(ii)] The coarse Ricci curvature bound provides TC inequalities for essentially any geometry, but it is only valid above a threshold temperature that depends on the locality of the Hamiltonian. Furthermore, such threshold temperature is in practice strictly larger than the true critical temperature (see \autoref{prop:curvaturebetac}).
    \item[(iii)] Quantum modified logarithmic Sobolev inequalities are typically more difficult to prove than their classical counterparts, and are currently only proven to hold in specific cases. However, for one-dimensional systems, a recently derived modified logarithmic Sobolev inequality \cite{bardet2021entropy,bardet2021rapid} provides us with TC (up to polylogarithmic overhead) at any positive temperature.
    \item[(iv)] The condition of local indistinguishability of the state for regular lattices. Although the condition can be checked for classical systems, we do not yet have a way to prove it in the quantum setting.
\end{itemize}

We conclude the article with two natural applications of our bounds. First, we derive Gaussian concentration bounds for a large class of Lipschitz observables whenever the state $\omega$ is that of a commuting Hamiltonian at large enough temperature (\autoref{concentration}). Second, we argue on the use of the transportation cost inequality in proving the equivalence between the microcanonical and the canonical ensembles and an exponential improvement over the weak Eigenstate Thermalization Hypothesis (\autoref{ETH}).

\section{Notations and basic definitions}

Given a finite set $V$, we denote by $\cH_V=\bigotimes_{v\in V}\cH_v$ the Hilbert space of $n=|V|$ qudits (\emph{i.e.}, $\cH_v\equiv \mathbb{C}^d$ for all $v\in V$) and by $\cB_V$ the algebra of linear operators on $\cH_V$. $\mathcal{O}_V$ corresponds to the space of self-adjoint linear operators on $\cH_V$, whereas $\mathcal{O}^T_V\subset \mathcal{O}_V$ is the subspace of traceless self-adjoint linear operators. $\mathcal{O}_V^+$ denotes the cone of positive semidefinite linear operators on $\cH_V$ and $\mathcal{S}_V\subset \mathcal{O}_V^+$ denotes the set of quantum states. We denote by $\cP_V$ the set of probability measures on $[d]^V$. For any subset $A\subseteq V$, we use the standard notations $\mathcal{O}_A, \mathcal{S}_A\ldots$ for the corresponding objects defined on subsystem $A$. Given a state $\rho\in\mathcal{S}_V$, we denote by $\rho_A$ its marginal onto the subsystem $A$. For any $X\in\mathcal{O}_V$, we denote by $\|X\|_1$ its trace norm. The identity on $\mathcal{O}_{v}$, $v\in V$, is denoted by $\mathbb{I}_v$.

Given two states $\rho,\,\omega\in \mathcal{S}_V$ such that $\operatorname{supp}(\rho)\subseteq \operatorname{supp}(\omega)$, their quantum relative entropy is defined as \cite{nielsen2002quantum,wilde2013quantum,holevo2012quantum}
\begin{align}
    S(\rho\|\omega)=\tr\big[\rho\,(\ln\rho-\ln\omega) \big]\,.
\end{align}
Whenever $\rho=\rho_{AB}$ is a bipartite state and $\omega=\rho_A\otimes\rho_B$, their relative entropy reduces to the mutual information 
\begin{align}
    I(A;B)_\rho:=S(\rho_{AB}\|\rho_A\otimes \rho_B)\,.
\end{align}
In the next sections, we also utilize the measured relative entropy \cite{donald1986relative,petz1986sufficient,hiai1991proper,berta2017variational}
\begin{align}
    S_{\mathbb{M}}(\rho\|\omega):=\sup_{(\mathcal{X},M)}\,S(P_{\rho,M}\|P_{\sigma,M})\,,
\end{align}
where the supremum above is over all positive operator valued measures $M$ that map the input quantum state to a probability distribution on a finite set $\mathcal{X}$ with probability mass function given by $P_{\rho,M}(x)=\tr\rho M(x)$.

In this paper, we study inequalities relating the $W_1$ distance between two states to their relative entropy. More precisely, for a fixed state $\omega\in \mathcal{S}_V$, we are interested in upper bounding the best constant $C(\omega)>0$ such that, for all $\rho\in \mathcal{S}_V$ with $\operatorname{supp}(\rho)\subseteq\operatorname{supp}(\omega)$, 
\begin{align}\label{eq:TC(c)}
    \left\|\rho-\omega\right\|_{W_1}\le \,\sqrt{C(\omega)\,S(\rho\|\omega)}\,.
\end{align}
In general, given a constant $c\ge C(\omega)$, we refer to the above inequality for $C(\omega)$ replaced by $c$ as a \textit{transportation cost inequality}, denoted by $\operatorname{TC}(c)$. As mentioned in the introduction, the following holds \cite[Theorem 2]{de2020quantum}:
\begin{prop}\label{prop:Martoniid}
For any product state $\omega\in\mathcal{S}_V$,
\begin{equation}
C(\omega) \le \frac{|V|}{2}\,.
\end{equation}
\end{prop}

In the next sections, we aim at recovering the linear dependence of the constant $C(\omega)$ on the size $n=|V|$ of the system under various measures of independence.

We will need the following properties of the quantum $W_1$ distance:

\begin{prop}[{\cite[Proposition 2]{de2020quantum}}]\label{prop2}
The quantum $W_1$ distance coincides with the trace distance for quantum states that differ in only one site, \emph{i.e.}, for any $X\in\mathcal{O}_V^T$ such that $\mathrm{Tr}_vX = 0$ for some $v\in V$ we have
\begin{equation}
\left\|X\right\|_{W_1} = \frac{1}{2}\left\|X\right\|_1\,.
\end{equation}
\end{prop}

\begin{prop}[{\cite[Proposition 5]{de2020quantum}}]\label{prop5}
The quantum $W_1$ distance between two quantum states that differ only in the region $A\subseteq V$ is at most $2\left|A\right|$ times their trace distance, \emph{i.e.}, for any $X\in\mathcal{O}_V^T$ such that $\mathrm{Tr}_AX=0$ we have
\begin{equation}
\left\|X\right\|_{W_1} \le \left|A\right|\left\|X\right\|_1\,.
\end{equation}
\end{prop}

\begin{prop}[Tensorization {\cite[Proposition 4]{de2020quantum}}]\label{prop4}
The quantum $W_1$ distance is additive with respect to the tensor product, \emph{i.e.}, let $A,\,B$ be disjoint subsets of $V$.
Then, for any $\rho_A,\,\sigma_A\in\mathcal{S}_A$ and any $\rho_B,\,\sigma_B\in\mathcal{S}_B$ we have
\begin{equation}
\left\|\rho_A\otimes\rho_B - \sigma_A\otimes\sigma_B\right\|_{W_1} = \left\|\rho_A - \sigma_A\right\|_{W_1} + \left\|\rho_B - \sigma_B\right\|_{W_1}\,.
\end{equation}
\end{prop}

\begin{prop}[{\cite[Proposition 13]{de2020quantum}}]\label{prop13}
Let $\Phi:\mathcal{O}_V\to\mathcal{O}_V$ be a quantum channel.
For any $v\in V$, let $A_v\subseteq V$ be the light-cone of the site $v$, \emph{i.e.}, the minimum subset of $V$ such that $\mathrm{Tr}_{A_i}\Phi(X) = 0$ for any $X\in\mathcal{O}_V$ such that $\mathrm{Tr}_vX=0$.
Then, $\Phi$ can expand the quantum $W_1$ distance by at most twice the size of the largest light-cone, \emph{i.e.}, for any $X\in\mathcal{O}_V^T$ we have
\begin{equation}
\left\|\Phi(X)\right\|_{W_1} \le 2\max_{v\in V}\left|A_v\right|\left\|X\right\|_{W_1}\,.
\end{equation}
\end{prop}

\begin{prop}[{\cite[Proposition 15]{de2020quantum}}]\label{prop15}
For any $H\in\mathcal{O}_V$ and any $v\in V$ we have
\begin{equation}
\left\|H - \mathbb{I}_v\otimes\frac{1}{d}\,\mathrm{Tr}_vH\right\|_\infty \le \left\|H\right\|_L\,.
\end{equation}
\end{prop}

\begin{prop}[{\cite[Corollary 1]{de2020quantum}}]\label{cor1}
For any $\rho,\,\sigma\in\mathcal{S}_V$,
\begin{equation}\label{eq:tproduct}
\left\|\rho - \sigma\right\|_{W_1} \ge \frac{1}{2}\sum_{v\in V}\left\|\rho_v - \sigma_v\right\|_1\,,
\end{equation}
and equality holds whenever both $\rho$ and $\sigma$ are product states.
\end{prop}

\begin{thm}[$W_1$ continuity of the entropy {\cite[Theorem 1]{de2020quantum}}]\label{thm1}
For any $\rho,\,\sigma\in\mathcal{S}_V$ we have
\begin{equation}
    \left|S(\rho) - S(\sigma)\right| \le g\left(\left\|\rho-\sigma\right\|_{W_1}\right) + \left\|\rho-\sigma\right\|_{W_1}\ln\left(d^2\left|V\right|\right)\,,
\end{equation}
where for any $t\ge0$
\begin{equation}
g(t) = \left(t+1\right)\ln\left(t+1\right) - t\ln t\,.
\end{equation}
\end{thm}

\section{Dobrushin uniqueness condition}\label{transportationcostdobrushin}
In this section, we consider a spin chain and prove the transportation cost inequality under a quantum generalization of Dobrushin's uniqueness condition \cite{dobrushin1970prescribing}.
Such condition is formulated in terms of the conditional probability distributions of the state of a subset of $V$ conditioned on the state of a second disjoint subset of $V$.
Therefore, formulating a quantum version of Dobrushin's uniqueness condition requires a quantum counterpart of the conditional probability distribution.
In the classical setting, given two random variables $X$ and $Y$ taking values in finite sets and with joint probability distribution $\omega_{XY}$, the conditional probability distribution $\omega_{Y|X}$ of $Y$ given $X$ with probability mass function
\begin{equation}
    \omega_{Y|X=x}(y) = \frac{\omega_{XY}(x,y)}{\omega_X(x)}
\end{equation}
represents the knowledge that we have on $Y$ when we know only the value of $X$.
We can associate to such conditional distribution the stochastic map $\Phi_{X\to XY}$ that has as input a probability distribution $p_X$ for $X$ and as output the joint probability distribution of $XY$ with probability mass function given by
\begin{equation}\label{eq:Phistoch}
    \Phi_{X\to XY}(p_X)(x,y) = \omega_{Y|X=x}(y)\,p_X(x) = \frac{\omega_{XY}(x,y)}{\omega_X(x)}\,p_X(x)\,. 
\end{equation}
In the quantum setting, we consider a bipartite quantum system $AB$ and a joint quantum state $\omega_{AB}$ of $AB$.
The quantum counterpart of the stochastic map \eqref{eq:Phistoch} is called quantum recovery map \cite{junge2018universal,sutter2017multivariate} and its action on a quantum state $\rho_A$ of $A$ is
\begin{equation}\label{eq:Phirecov}
\Phi_{A\to AB}(\rho_A) = \int_\mathbb{R}\omega_{AB}^{\frac{1-it}{2}}\,\omega_A^\frac{it-1}{2}\,\rho_A\,\omega_A^{-\frac{1+it}{2}}\,\omega_{AB}^{\frac{1+it}{2}}\,d\mu_0(t)\,,
\end{equation}
where $\mu_0$ is the probability distrbution on $\mathbb{R}$ with density
\begin{equation}
    d\mu_0(t) = \frac{\pi\,dt}{2\left(\cosh(\pi t)+1\right)}\,.
\end{equation}
We stress that \eqref{eq:Phirecov} reduces to \eqref{eq:Phistoch} whenever $\rho_A$, $\omega_A$ and $\omega_{AB}$ commute.
If $A$ is in the state $\omega_A$, the recovery map $\Phi_{A\to AB}$ recovers the joint state $\omega_{AB}$, \emph{i.e.}, $\Phi_{A\to AB}(\omega_A) = \omega_{AB}$.
The relevance of the recovery map comes from the recoverability theorem \cite{sutter2017multivariate}, which states that $\Phi_{A\to AB}$ can recover a generic joint state $\rho_{AB}$ from its marginal $\rho_A$ if removing the subsystem $B$ does not significantly decrease the relative entropy between $\rho$ and $\omega$.
More precisely, for any quantum state $\rho_{AB}$ of $AB$ we have
\begin{equation}\label{eq:recovery}
S(\rho_{AB}\|\omega_{AB}) - S(\rho_A\|\omega_A) \ge S_{\mathbb{M}}(\rho_{AB}\|\Phi_{A\to AB}(\rho_A))\,.
\end{equation}

We consider the setting where $V$ is partitioned as
\begin{equation}
    V = A_1\sqcup\ldots\sqcup A_m\,.
\end{equation}
For any $i\in[m]$, we denote with $A_1^i$ the union $A_1\sqcup\ldots\sqcup A_i$.
The recoverability theorem implies the following \autoref{lem:recov}, which we will employ several times:
\begin{lem}\label{lem:recov}
For any $\rho,\,\omega\in\mathcal{S}_V$ we have
\begin{align}\label{eq:ST1}
    &S(\rho\|\omega) \ge \frac{1}{2m}\left(\sum_{i=1}^m\left\|\rho_{A_1^i} - \Phi_{A_1^{i-1}\to A_1^i}(\rho_{A_1^i})\right\|_1\right)^2\,,\\
    &\left(\frac{1}{2m}\sum_{i=1}^m\left\|\rho_{A_1^i} - \Phi_{A_1^{i-1}\to A_1^i}(\rho_{A_1^i})\right\|_1\right)^2 \le 1-\exp\left(-\frac{S(\rho\|\omega)}{m}\right)\,,\label{eq:ST2}
\end{align}
where $\Phi_{A_1^{i-1}\to A_1^i}$ are the recovery maps associated to $\omega$.
\end{lem}
\begin{proof}
Eq. \eqref{eq:Phirecov} and Pinsker's inequality imply for any $i\in[m]$
\begin{equation}\label{eq:STi}
    S(\rho_{A_1^i}\|\omega_{A_1^i}) - S(\rho_{A_1^{i-1}}\|\omega_{A_1^{i-1}}) \ge S_\mathbb{M}\left(\rho_{A_1^i}\left\|\Phi_{A_1^{i-1}\to A_1^i}(\rho_{A_1^{i-1}})\right.\right) \ge \frac{1}{2}\left\|\rho_{A_1^i} - \Phi_{A_1^{i-1}\to A_1^i}(\rho_{A_1^{i-1}})\right\|_1^2\,.
\end{equation}
Summing \eqref{eq:STi} over $i$ and using the convexity of the square function yields
\begin{equation}
    S(\rho\|\omega) \ge \frac{1}{2}\sum_{i=1}^m\left\|\rho_{A_1^i} - \Phi_{A_1^{i-1}\to A_1^i}(\rho_{A_1^{i-1}})\right\|_1^2 \ge \frac{1}{2m}\left(\sum_{i=1}^m\left\|\rho_{A_1^i} - \Phi_{A_1^{i-1}\to A_1^i}(\rho_{A_1^i})\right\|_1\right)^2\,.
\end{equation}
The claim \eqref{eq:ST1} follows.

With an analogous proof, applying the improved Pinsker's inequality
\begin{equation}
    \frac{1}{4}\left\|\sigma - \tau\right\|_1^2 \le 1 - e^{-S_\mathbb{M}(\sigma\|\tau)}
\end{equation}
and the convexity of the function $t\mapsto -\ln\left(1-t^2\right)$ we get
\begin{equation}
    S(\rho\|\omega) \ge -m\ln\left(1-\left(\frac{1}{2m}\sum_{i=1}^m\left\|\rho_{A_1^i} - \Phi_{A_1^{i-1}\to A_1^i}(\rho_{A_1^i})\right\|_1\right)^2\right)\,.
\end{equation}
The claim \eqref{eq:ST2} follows.
\end{proof}
The following property of the recovery map will be fundamental:
\begin{lem}\label{lem:Philoc}
Let $\omega_{ABC}$ be a joint state of the tripartite quantum system $ABC$.
Let us assume that $\omega_{ABC}$ is Markovian, \emph{i.e.},
\begin{equation}
    I(A;C|B)_\omega = 0\,.
\end{equation}
Then, the recovery map $\Phi_{AB\to ABC}$ associated to $\omega_{ABC}$ does not act on the subsystem $A$.
\begin{proof}
From the characterization of the states that saturate the strong subadditivity \cite{hayden2004structure}, the Hilbert space $\mathcal{H}_B$ of $B$ has a decomposition
\begin{equation}
    \mathcal{H}_B = \bigoplus_{i=1}^k \mathcal{H}_{B_i^L}\otimes\mathcal{H}_{B_i^R}\,,
\end{equation}
where the Hilbert spaces $\left\{\mathcal{H}_{B_i^L}\right\}_{i=1}^k$ and $\left\{\mathcal{H}_{B_i^L}\right\}_{i=1}^k$ are pairwise orthogonal, and $\omega_{ABC}$ can be expressed as
\begin{equation}
    \omega_{ABC} = \bigoplus_{i=1}^k p_i\,\omega_{AB_i^L}^{(i)}\otimes\omega_{B_i^RC}^{(i)}\,,
\end{equation}
where $p$ is a probability distribution on $[k]$, and each $\omega_{AB_i^L}^{(i)}$ or $\omega_{B_i^RC}^{(i)}$ is a quantum state with support in the corresponding $\mathcal{H}_A\otimes\mathcal{H}_{B_i^L}$ or $\mathcal{H}_{B_i^R}\otimes\mathcal{H}_C$.
We have for any quantum state $\rho_{AB}$ of $AB$
\begin{equation}
\Phi_{AB\to ABC}(\rho_{AB}) = \int_\mathbb{R}\omega_{ABC}^{\frac{1-it}{2}}\,\omega_{AB}^\frac{it-1}{2}\,\rho_{AB}\,\omega_{AB}^{-\frac{1+it}{2}}\,\omega_{ABC}^{\frac{1+it}{2}}\,d\mu_0(t)\,.
\end{equation}
We have for any $z\in\mathbb{C}$ that
\begin{equation}
    \omega_{ABC}^z\,\omega_{AB}^{-z} = \bigoplus_{i=1}^k\left(\omega_{B_i^RC}^{(i)}\right)^z\left(\omega_{B_i^R}^{(i)}\right)^{-z}
\end{equation}
does not act on $A$, and the claim follows choosing $z = \left(1-it\right)/2$.
\end{proof}
\end{lem}

\subsection{Markovian case}
In this subsection, we assume that $\omega\in\mathcal{S}_V$ is a one-dimensional quantum Markov state.
More preciesly, let $\left\{A_1,\,\ldots,\,A_m\right\}$ be a partition of $V$ and let $K = \max\left(\left|A_1\right|,\,\ldots,\,\left|A_m\right|\right)$.
Then, we assume that
\begin{equation}\label{eq:markov2}
I(A_i;A_1^{i-2}|A_{i-1})_\omega = 0
\end{equation}
for any $i=3,\,\ldots,\,m$.
For any $i\in[m]$, let $\Phi_i$ be the recovery map \eqref{eq:Phirecov} associated to $\omega_{A_1^i}$ that recovers $A_i$ from $A_1^{i-1}$.
From \autoref{lem:Philoc}, $\Phi_i$ acts only on $A_{i-1}$, \emph{i.e.}, it is a map $\Phi_i:\mathcal{O}_{A_{i-1}}\to\mathcal{O}_{A_{i-1}A_i}$.
We also define
\begin{equation}
    \tilde{\Phi}_i = \mathrm{Tr}_{A_{i-1}}\circ\Phi_i:\mathcal{O}_{A_{i-1}}\to\mathcal{O}_{A_i}\,.
\end{equation}

We can now state the main result of this Section:
\begin{thm}\label{thm:Markov}
Let us assume that for any $i\in[m]$, $\tilde{\Phi}_i$ is a contraction with respect to the trace norm for all the couples of quantum states of $A_1^{i-1}$ that differ only on the subsystem $A_{i-1}$, \emph{i.e.}, that coincide after discarding $A_{i-1}$.
More precisely, we assume that there exists $0\le\eta<1$ such that for any $i\in[m]$ and any $X\in\mathcal{O}_{A_1^{i-1}}^T$ with $\mathrm{Tr}_{A_{i-1}}X=0$ we have
\begin{equation}\label{eq:eta}
    \left\|\tilde{\Phi}_i(X)\right\|_1 \le \eta\left\|X\right\|_1\,.
\end{equation}
Then, we have
\begin{equation}\label{eq:marton2}
C(\omega) \le 2m\,K^2\left(\frac{1}{1-\eta}+1\right)^2\,.
\end{equation}
Furthermore, for any $\rho\in\mathcal{S}_V$ we have
\begin{equation}\label{eq:marton2a}
\left\|\rho - \omega\right\|_{W_1} \le K\left(\frac{1}{1-\eta}+1\right)2m\sqrt{1-e^{-\frac{S(\rho\|\omega)}{m}}}\,.
\end{equation}
\begin{proof}
Let $\rho\in\mathcal{S}_V$.
On the one hand, we have from \autoref{lem:W1Phi2}
\begin{align}
\left\|\rho - \omega\right\|_{W_1} &= \left\|\sum_{i=1}^m(\Phi_m\circ\ldots\circ \Phi_{i+1})(\rho_{A_1^i} - \Phi_i(\rho_{A_1^{i-1}}))\right\|_{W_1}\nonumber\\
&\le \sum_{i=1}^m\left\|(\Phi_m\circ\ldots\circ \Phi_{i+1})(\rho_{A_1^i} - \Phi_i(\rho_{A_1^{i-1}}))\right\|_{W_1}\nonumber\\
&\le K\left(\frac{1}{1-\eta}+1\right)\sum_{i=1}^m\left\|\rho_{A_1^i} - \Phi_i(\rho_{A_1^{i-1}})\right\|_1\,.
\end{align}
On the other hand, we have from \eqref{eq:ST1} of \autoref{lem:recov}
\begin{equation}\label{eq:pinsker}
S(\rho\|\omega) \ge \frac{1}{2m}\left(\sum_{i=1}^m\left\|\rho_{A_1^i} - \Phi_i(\rho_{A_1^{i-1}})\right\|_1\right)^2\,,
\end{equation}
and the claim \eqref{eq:marton2} follows.
The claim \eqref{eq:marton2a} follows by employing \eqref{eq:ST2} in place of \eqref{eq:ST1}.
\end{proof}
\end{thm}
\begin{rem}\label{rem:eta}
Condition \eqref{eq:eta} holds for some $\eta<1$ iff $\tilde{\Phi}_i$ strictly decreases the trace distance between any two quantum states that differ only in the subsystem $A_i$ on which $\tilde{\Phi}_i$ acts, \emph{i.e.}, that coincide after discarding $A_i$.
We expect this condition to hold for any strictly positive temperature.
\end{rem}
\begin{rem}\label{rem:1D}
An example of quantum state satisfying \eqref{eq:markov2} is a Gibbs state of a nearest-neighbor Hamiltonian on the $D$-dimensional cubic lattice $\Lambda = [L]^D$, where $x,\,y\in\Lambda$ are neighbors iff $\left\|x-y\right\|_1=1$.
We can then choose $m = L + 1$ and
\begin{equation}
    A_i = \left\{x\in\Lambda: x_1 = i\right\}\,,\qquad i=0,\,\ldots,\,L\,,
\end{equation}
with
\begin{equation}
    K = \left(L+1\right)^{D-1}\,,
\end{equation}
and get from \autoref{thm:Markov}
\begin{equation}
C(\omega) \le 2\left(L+1\right)^{2D-1}\left(\frac{1}{1-\eta}+1\right)^2 = 2\left|V\right|^\frac{2D-1}{D}\left(\frac{1}{1-\eta}+1\right)^2\,.
\end{equation}
We stress that, assuming that $\eta$ remains bounded away from $1$, we get $C(\omega) = O(|V|)$ iff $D=1$, \emph{i.e.}, for one-dimensional systems.
\end{rem}

\begin{rem}\label{rem:eta}
We can choose
\begin{align}
    \eta &= \max\left\{\left\|\tilde{\Phi}_i(X) - \omega_{A_i}\otimes \mathrm{Tr}_{A_{i-1}}X\right\|_1:i\in[m]\,,\;X\in\mathcal{O}_{A_1^{i-1}}\,,\;\left\|X\right\|_1=1\right\}\nonumber\\
    &\le \max_{i\in[m]}\left\|\tilde{\Phi}_i - \omega_{A_i}\otimes\mathrm{Tr}_{A_{i-1}}\right\|_\diamond\,,
\end{align}
where $\omega_{A_i}\otimes\mathrm{Tr}_{A_{i-1}}:\mathcal{O}_{A_{i-1}}\to\mathcal{O}_{A_i}$ is the quantum channel that replaces the input quantum state with $\omega_{A_i}$ and
\begin{equation}
\left\|\Phi\right\|_\diamond = \sup\left\{\left\|(\Phi\otimes\mathbb{I}_{\mathcal{B}(\mathcal{H})})(X)\right\|_1:X\in\mathcal{B}(\mathcal{H}^{\otimes2})\,,\;\left\|X\right\|_1=1\right\}
\end{equation}
denotes the diamond norm of the linear map $\Phi$ on $\mathcal{B}(\mathcal{H})$. 
\end{rem}

\begin{prop}\label{propcondition}
Let $\omega\in\mathcal{S}_V$ satisfy \eqref{eq:markov2}, and assume
\begin{equation}\label{eq:Dinf}
a = \max_{i\in[m-1]} S_\infty(\omega_{A_i}\otimes\omega_{A_{i+1}}\|\omega_{A_i A_{i+1}})<\frac{1}{2}\,,
\end{equation}
where
\begin{equation}
    S_\infty(\rho\|\sigma) = \ln\inf\left\{\lambda\in\mathbb{R} : \rho \le \lambda\,\sigma\right\}
\end{equation}
denotes the quantum max-divergence \cite{tomamichel2015quantum} between the quantum states $\rho$ and $\sigma$.
Then, we can choose in \eqref{eq:eta}
\begin{equation}
\eta = \sqrt{2\,a}\,.
\end{equation}
\begin{proof}
From \autoref{rem:eta}, we can choose
\begin{equation}
    \eta = \max_{i\in[m]}\max_{|\psi_i\rangle}\left\|\tilde{\Phi}_i(|\psi_i\rangle\langle\psi_i|) - \omega_{A_i}\otimes\mathrm{Tr}_{A_{i-1}}|\psi_i\rangle\langle\psi_i|\right\|_1 \le \max_{i\in[m]}\max_{|\psi_i\rangle}\left\|\Phi_i(|\psi_i\rangle\langle\psi_i|) - \omega_{A_i}\otimes|\psi_i\rangle\langle\psi_i|\right\|_1\,,
\end{equation}
where each $|\psi_i\rangle$ is a unit vector in $\mathcal{H}_{A_1^{i-1}}$.
We have from Pinsker's inequality
\begin{equation}
\left\|\Phi_i(|\psi_i\rangle\langle\psi_i|) - |\psi_i\rangle\langle\psi_i|\otimes\omega_{A_i}\right\|_1 \le \sqrt{2\,S_{\mathbb{M}}(|\psi_i\rangle\langle\psi_i|\otimes\omega_{A_i}\|\Phi_i(|\psi_i\rangle\langle\psi_i|))}\,.
\end{equation}
\eqref{eq:Dinf} implies
\begin{equation}\label{eq:prod3}
\ln\omega_{A_{i-1}A_i} \ge \ln\omega_{A_{i-1}} + \ln\omega_{A_i} - a\,.
\end{equation}
From the characterization of the states that saturate the strong subadditivity \cite{hayden2004structure} we get
\begin{equation}
    \ln\omega_{A_1^{i-1}} + \ln\omega_{A_{i-1}A_i}  = \ln\omega_{A_{i-1}} + \ln\omega_{A_1^i}\,,
\end{equation}
therefore, \eqref{eq:prod3} can be rewritten as
\begin{equation}\label{eq:prod4}
\ln\omega_{A_1^i} \ge \ln\omega_{A_1^{i-1}} + \ln\omega_{A_i} - a\,.
\end{equation}
Choosing in \eqref{eq:STi} $\rho_{A_1^i} = |\psi_i\rangle\langle\psi_i|\otimes\omega_{A_i}$ we get with the help of \eqref{eq:prod4}
\begin{equation}
S_{\mathbb{M}}(|\psi_i\rangle\langle\psi_i|\otimes\omega_{A_i}\|\Phi_i(|\psi_i\rangle\langle\psi_i|)) \le \langle\psi_i|\left(\ln\omega_{A_1^{i-1}} - \mathrm{Tr}_{A_i}\left[\omega_{A_i}\ln\omega_{A_1^i}
\right]\right)|\psi_i\rangle - S(A_i)_\omega \le a\,,
\end{equation}
and the claim follows.
\end{proof}
\end{prop}

\begin{rem}
Condition \eqref{eq:eta} is reminiscent of the so-called Dobrushin uniqueness condition (see \cite[Theorem 4]{dobrushin1970prescribing}). 
\end{rem}

\subsection{Non-Markovian states}

Here, we prove an alternative version of \autoref{thm:Markov} where the Markov condition \eqref{eq:markov2} is replaced by exponential decay of correlations.

\begin{thm}\label{thm:general}
Let $V=[n]$ be a one-dimensional lattice, and let $\omega\in\mathcal{S}_V$.
For any $i\in[n]$, let $\Phi_i$ be the recovery map associated to $\omega_{1\ldots i}$ that recovers the site $i$ from the sites $1\ldots i-1$.
We assume that $\omega$ has exponentially decaying correlations, in the sense that there exist $C\ge0$ and $0\le\eta<1$ such that for any $i\in[n]$, any $k=0,\,\ldots,\,\max(i,\,n-i)$ and any $\tau\in\mathcal{S}_{1\ldots i}$,
\begin{equation}\label{eq:recovery4}
    \left\|\mathrm{Tr}_{i-k+1\ldots i+k}(\Phi_n\circ\ldots\circ\Phi_{i+1})(\tau_{1\ldots i}) - \tau_{1\ldots i-k}\otimes\omega_{i+k+1\ldots n}\right\|_1 \le C\,\eta^k\,.
\end{equation}
We also assume that for any $i\in[n]$, any $k=0,\,\ldots,\,i-1$ and any $\tau\in\mathcal{S}_{1\ldots i-1}$
\begin{equation}\label{eq:recovery5}
 \left\|\mathrm{Tr}_{i-k\ldots i}\Phi_i(\tau_{1\ldots i-1}) - \tau_{1\ldots i-k-1}\right\|_1 \le C\,\eta^k\,.
\end{equation}
Then,
\begin{equation}
C(\omega) \le 8\,n\left(2+\frac{C+1}{1-\eta}-\frac{\ln\left(C^2\,n\right)}{2\ln\eta}\right)^2\,.
\end{equation}

\begin{proof}
From \eqref{eq:ST1} of \autoref{lem:recov}, we have for any $\rho\in\mathcal{S}_V$
\begin{equation}
S(\rho\|\omega) \ge \frac{1}{2n}\left(\sum_{i=1}^n\left\|\rho_{1\ldots i} - \Phi_i(\rho_{1\ldots i-1})\right\|_1\right)^2\,.
\end{equation}
We have
\begin{align}
\left\|\rho - \omega\right\|_{W_1} &= \left\|\sum_{i=1}^n(\Phi_n\circ\ldots\circ \Phi_{i+1})(\rho_{1\ldots i} - \Phi_i(\rho_{1\ldots i-1}))\right\|_{W_1}\nonumber\\
&\le \sum_{i=1}^n\left\|(\Phi_n\circ\ldots\circ \Phi_{i+1})(\rho_{1\ldots i} - \Phi_i(\rho_{1\ldots i-1}))\right\|_{W_1}\,.
\end{align}
For any $i\in[n]$, we have from \autoref{lem:XAit}
\begin{align}
    &\left\|(\Phi_n\circ\ldots\circ \Phi_{i+1})(\rho_{1\ldots i} - \Phi_i(\rho_{1\ldots i-1}))\right\|_{W_1}\nonumber\\
    &\le 2\sum_{k=0}^{\max\{i,\,n-i\}}\left\|\mathrm{Tr}_{i-k+1\ldots i+k}(\Phi_n\circ\ldots\circ \Phi_{i+1})(\rho_{1\ldots i} - \Phi_i(\rho_{1\ldots i-1}))\right\|_1\nonumber\\
    &\le 2\sum_{k=0}^{\max\{i,\,n-i\}}\left(\left\|\rho_{1\ldots i-k} - \mathrm{Tr}_{i-k+1\ldots i}\Phi_i(\rho_{1\ldots i-1}))\right\|_1 + C\,\eta^k\left\|\rho_{1\ldots i} - \Phi_i(\rho_{1\ldots i-1}))\right\|_1\right)\,.
\end{align}
We have for any $k_0\in\left\{0,\,\ldots,\,i-1\right\}$
\begin{align}
    &\sum_{k=0}^{i-1}\left\|\rho_{1\ldots i-k} - \mathrm{Tr}_{i-k+1\ldots i}\Phi_i(\rho_{1\ldots i-1})\right\|_1\nonumber\\
    &\le \sum_{k=0}^{k_0}\left\|\rho_{1\ldots i} - \Phi_i(\rho_{1\ldots i-1})\right\|_1 + \sum_{k=k_0+1}^{i-1}\left\|\rho_{1\ldots i-k} - \omega_{1\ldots i-k} - \mathrm{Tr}_{i-k+1\ldots i}\Phi_i(\rho_{1\ldots i-k} - \omega_{1\ldots i-k})\right\|_1\nonumber\\
    &\le \left(k_0+1\right)\left\|\rho_{1\ldots i} - \Phi_i(\rho_{1\ldots i-1})\right\|_1 + C\sum_{k=k_0+1}^{i-1}\eta^{k-1}\left\|\rho_{1\ldots i-k} - \omega_{1\ldots i-k}\right\|_1\nonumber\\
    &\le \left(k_0+1\right)\left\|\rho_{1\ldots i} - \Phi_i(\rho_{1\ldots i-1})\right\|_1 + \frac{C\,\eta^{k_0}}{1-\eta}\left\|\rho-\omega\right\|_1\,,
\end{align}
therefore
\begin{align}
    &\left\|(\Phi_n\circ\ldots\circ \Phi_{i+1})(\rho_{1\ldots i} - \Phi_i(\rho_{1\ldots i-1}))\right\|_{W_1}\nonumber\\
    &\le 2\left(\left(k_0+1+\frac{C}{1-\eta}\right)\left\|\rho_{1\ldots i} - \Phi_i(\rho_{1\ldots i-1})\right\|_1 + \frac{C\,\eta^{k_0}}{1-\eta}\left\|\rho-\omega\right\|_1\right)\,,
\end{align}
and
\begin{align}
    \left\|\rho - \omega\right\|_{W_1} &\le 2\left(\left(k_0+1+\frac{C}{1-\eta}\right)\sum_{i=1}^n\left\|\rho_{1\ldots i} - \Phi_i(\rho_{1\ldots i-1})\right\|_1 + \frac{n\,C\,\eta^{k_0}}{1-\eta}\left\|\rho-\omega\right\|_1\right)\nonumber\\
    &\le 2\left(k_0+1+C\,\frac{1+\eta^{k_0}\sqrt{n}}{1-\eta}\right)\sqrt{2\,n\,S(\rho\|\omega)}\,,
\end{align}
and the claim follows choosing
\begin{equation}
    k_0 = \left\lceil-\frac{\ln\left(C^2\,n\right)}{2\ln\eta}\right\rceil\,.
\end{equation}
\end{proof}
\end{thm}

\subsection{Auxiliary lemmas}
\begin{lem}\label{lem:W1Phi2}
Under the same hypotheses of \autoref{thm:Markov}, for any $i\in[m]$ and any $X\in\mathcal{O}_{A_1^i}^T$ such that
\begin{equation}\label{eq:hypX}
    \mathrm{Tr}_{A_{i-1}A_i}X=0
\end{equation}
we have
\begin{equation}
    \left\|(\Phi_m\circ\ldots\circ\Phi_{i+1})(X)\right\|_{W_1} \le K\left(\frac{1}{1-\eta}+1\right)\left\|X\right\|_1\,.
\end{equation}
\begin{proof}
We have from \autoref{lem:XA} and from the contractivity of the trace distance
\begin{align}
    \left\|(\Phi_m\circ\ldots\circ\Phi_{i+1})(X)\right\|_{W_1} &\le \left|A_{i-1}\right|\left\|(\Phi_m\circ\ldots\circ\Phi_{i+1})(X)\right\|_1  + \left\|(\Phi_m\circ\ldots\circ\Phi_{i+1})\left(\mathrm{Tr}_{A_{i-1}}X\right)\right\|_{W_1}\nonumber\\
    &\le K\left\|X\right\|_1  + \left\|(\Phi_m\circ\ldots\circ\Phi_{i+1})\left(\mathrm{Tr}_{A_{i-1}}X\right)\right\|_{W_1}\,.
\end{align}
We have
\begin{align}
&\left\|(\Phi_m\circ\ldots\circ\Phi_{i+1})\left(\mathrm{Tr}_{A_{i-1}}X\right)\right\|_{W_1}\nonumber\\
&\le \left|A_i\right|\left\|(\Phi_m\circ\ldots\circ\Phi_{i+1})\left(\mathrm{Tr}_{A_{i-1}}X\right)\right\|_1 + \left\|\mathrm{Tr}_{A_i}(\Phi_m\circ\ldots\circ\Phi_{i+1})\left(\mathrm{Tr}_{A_{i-1}}X\right)\right\|_{W_1}\nonumber\\
&\le K\left\|\mathrm{Tr}_{A_{i-1}}X\right\|_1 + \left\|(\Phi_m\circ\ldots\circ\Phi_{i+2}\circ\tilde{\Phi}_{i+1})\left(\mathrm{Tr}_{A_{i-1}}X\right)\right\|_{W_1}\,.
\end{align}
Iterating the procedure we get
\begin{align}
&\left\|(\Phi_m\circ\ldots\circ\Phi_{i+1})\left(\mathrm{Tr}_{A_{i-1}}X\right)\right\|_{W_1}\nonumber\\
&\le K\left(\left\|\mathrm{Tr}_{A_{i-1}}X\right\|_1 + \left\|\tilde{\Phi}_{i+1}\left(\mathrm{Tr}_{A_{i-1}}X\right)\right\|_1 + \ldots + \left\|(\tilde{\Phi}_m\circ\ldots\circ\tilde{\Phi}_{i+1})\left(\mathrm{Tr}_{A_{i-1}}X\right)\right\|_1\right)\nonumber\\
& \phantom{\le} + \left\|\mathrm{Tr}_{A_m}(\tilde{\Phi}_m\circ\ldots\circ\tilde{\Phi}_{i+1})\left(\mathrm{Tr}_{A_{i-1}}X\right)\right\|_{W_1}\nonumber\\
&\le K\left(1+\eta+\ldots + \eta^{m-i}\right)\left\|\mathrm{Tr}_{A_{i-1}}X\right\|_1 + \left\|\mathrm{Tr}_{A_{i-1}A_i}X\right\|_1\nonumber\\
&\le \frac{K}{1-\eta}\left\|\mathrm{Tr}_{A_{i-1}}X\right\|_1 \le \frac{K}{1-\eta}\left\|X\right\|_1\,,
\end{align}
where the last two inequalities follow from \eqref{eq:eta} and \eqref{eq:hypX}, respectively.
The claim follows.
\end{proof}
\end{lem}

\begin{lem}\label{lem:XA}
For any $X\in\mathcal{O}_V^T$ and any $A\subseteq V$,
\begin{equation}
\left\|X\right\|_{W_1} \le \left|A\right|\left\|X\right\|_1 + \left\|\mathrm{Tr}_AX\right\|_{W_1}\,.
\end{equation}
\begin{proof}
Without loss of generality, we can assume that $V=[n]$ and $A = [k]$ for some $k\in[n]$.
We have
\begin{align}
\left\|X\right\|_{W_1} &\le \left\|X - \frac{\mathbb{I}}{d}\otimes\mathrm{Tr}_1X\right\|_{W_1} + \left\|\frac{\mathbb{I}}{d}\otimes\mathrm{Tr}_1X\right\|_{W_1} = \frac{1}{2}\left\|X - \frac{\mathbb{I}}{d}\otimes\mathrm{Tr}_1X\right\|_1 + \left\|\mathrm{Tr}_1X\right\|_{W_1}\nonumber\\
&\le \left\|X\right\|_1 + \left\|\mathrm{Tr}_1X\right\|_{W_1}\,,
\end{align}
where the equality follows from \autoref{prop2} and \autoref{prop4} and the last inequality follows from the triangle inequality for the trace norm and its contractivity with respect to partial traces.
By induction we get
\begin{align}
\left\|X\right\|_{W_1} &\le \left(\left\|X\right\|_1 + \ldots + \left\|\mathrm{Tr}_{1\ldots k-1}X\right\|_1\right) + \left\|\mathrm{Tr}_{1\ldots k}X\right\|_{W_1}\nonumber\\
&\le k\left\|X\right\|_1 + \left\|\mathrm{Tr}_{1\ldots k}X\right\|_{W_1}\,,
\end{align}
and the claim follows.
\end{proof}
\end{lem}

\begin{lem}\label{lem:XAit}
Let $V=[n]$.
Then, for any $X\in\mathcal{O}_V^T$,
\begin{equation}
    \left\|X\right\|_{W_1} \le \left\|X\right\|_1 + \left\|\mathrm{Tr}_1X\right\|_1 + \ldots + \left\|\mathrm{Tr}_{1\ldots n-1}X\right\|_1\,.
\end{equation}
\begin{proof}
Follows from \autoref{lem:XA}.
\end{proof}
\end{lem}

\section{Curvature  bound}\label{Lipcontract}

In the seminal paper \cite{Ollivier2009}, Ollivier introduced a generalization of the notion of curvature to generic, possibly discrete, metric spaces. In his framework, the curvature of a metric space $(\Omega,d)$ endowed with a classical stochastic map $P$ acting on the probability measures on $\Omega$ is defined as the following contraction property of the Wasserstein distance $W_1$: for any two probability measures $\mu_1,\mu_2$, 
\begin{align}\label{Olliviercurvature}
   W_1(P(\mu_1),P(\mu_2))\le \big(1-\kappa\big)\,W_1(\mu_1,\mu_2)\,. 
\end{align}
The constant $\kappa>0$ is called the coarse Ricci curvature of the triple $(\Omega,d,P)$. In particular, it is easy to verify that the existence of a positive coarse Ricci curvature induces the uniqueness of the invariant measure $\nu$ for the Markov kernel $P$. Moreover, it was recently proven in \cite{Eldan2017} that Ollivier's coarse Ricci curvature provides an upper bound on the transportation cost inequality for the measure $\nu$, hence recovering the results from the smooth Riemannian setting. 

Here, inspired by the works of \cite{Ollivier2009} and \cite{Eldan2017}, we prove that a contraction of the Lipschitz constant under a certain quantum channel constructed from the Petz recovery maps of the Gibbs state $\omega$ can be used to conclude that $\omega$ satisfies a transportation cost inequality. In particular, we do not need to assume that the underlying graph is $\mathbb{Z}$, in contrast with \autoref{transportationcostdobrushin}. Let $G=(V,E)$ be a hypergraph with $n=|V|$, and let $H:=\sum_{A\in E}h_A$ be a Hamiltonian whose local terms $h_A$ pairwise commute and are supported on the hyperedges $A\in E$. For a given site $v\in V$, we recall the composition of the partial trace $\tr_v$ on $v$ with the rotated Petz recovery map of $v$:
\begin{align}
    \Psi_v(\rho)=\Phi_v\circ \tr_v(\rho)=\int_{\mathbb{R}}\omega^{\frac{1-it}{2}}\omega_{v^c}^{\frac{-1+it}{2}}\,(\rho_{v^c}\otimes I_v)\,\omega_{v^c}^{\frac{-1-it}{2}}\omega^{\frac{1+it}{2}}\,d\mu_0(t)
\end{align}
for the probability density $\mu_0(t):=\frac{\pi}{2}\big(\cosh(\pi t)+1\big)^{-1}$. Note that since we assumed $\omega$ to be the Gibbs state of a commuting Hamiltonian, the map $\Psi_v$ acts non-trivially on the neighborhood of $v$
\begin{equation}
    N_v:=\bigcup\left\{A\in E : v\in A\right\}\,.
\end{equation}
We also introduce the quantum channel
\begin{equation}
    \Psi=\frac{1}{n}\sum_{v\in V}\Psi_v\,.
\end{equation}
We assume that $\Psi$ is a contraction with respect to the $W_1$ norm, \emph{i.e.}, that
\begin{align}\label{contraction2}
    \left\|\Psi\right\|_{W_1\to W_1} = \max_{\Delta\in\mathcal{O}_V^T}\frac{\left\|\Psi(\Delta)\right\|_{W_1}}{\left\|\Delta\right\|_{W_1}} \le 1-\frac{\kappa}{n}
\end{align}
for some $\kappa>0$, in analogy with \eqref{Olliviercurvature}. This contraction property was already derived in Ollivier's original article \cite{Ollivier2009} as a generalization of Dobrushin's uniqueness condition. Here, we first prove that this condition implies the transportation cost inequality for the Gibbs state $\omega\equiv \omega_\beta:=e^{-\beta H}/\tr\,e^{-\beta H}$:

\begin{thm}\label{thmTCcurvature}
With the conditions of the previous paragraph, we have
\begin{align}
    C(\omega_\beta) \le 2n\,\frac{N^2}{\left(1-e^{-\kappa}\right)^2}\,,
\end{align}
where $N:=\max_{v\in V}|N_v|$.
\end{thm}

\begin{proof}
We have for any state $\rho\in\mathcal{S}_V$
\begin{align}\label{equationsum2}
    \left\|\rho - \omega_\beta\right\|_{W_1}\le \sum_{i=1}^n \left\|\Psi^{i-1}(\rho)-\Psi^i(\rho)\right\|_{W_1}+\left\|\Psi^n(\rho)-\omega_\beta\right\|_{W_1}\,.
\end{align}
The last term can be controlled by $\left\|\rho-\omega_\beta\right\|_{W_1}$ thanks to the contraction \eqref{contraction2}: 
\begin{align}\label{eq:1step2}
   \left\|\Psi^n(\rho)-\omega_\beta\right\|_{W_1}\le \left(1-\frac{\kappa}{n}\right)^n\left\|\rho-\omega_\beta\right\|_{W_1}\le  e^{-\kappa}\left\|\rho-\omega_\beta\right\|_{W_1}\,.
\end{align}
On the other hand, the sum on the right-hand side of \eqref{equationsum2} can be controlled as follows:
\begin{align}
    \sum_{i=1}^n\left\|\Psi^{i-1}(\rho)-\Psi^i(\rho)\right\|_{W_1}&\le\frac{1}{n}\sum_{i=1}^n\sum_{v\in V} \left\|\Psi_v(\Psi^{i-1}(\rho))-\Psi^{i-1}(\rho)\right\|_{W_1}\nonumber\\
    &\le  \frac{N}{n} \sum_{i=1}^n\sum_{v\in V}\left\|\Psi_v(\Psi^{i-1}(\rho))-\Psi^{i-1}(\rho)\right\|_1\,,
\end{align}
where the last inequality follows by \autoref{prop5}.
Proceeding as in the proof of \autoref{lem:recov}, by the joint use of Pinsker's inequality with the recoverability bound followed by the data processing inequality we can further bound the trace distances above so that
\begin{align}
     \sum_{i=1}^n\left\|\Psi^{i-1}(\rho)-\Psi^i(\rho)\right\|_{W_1}&\le \frac{N}{n}\,\sum_{i=1}^n\sum_{v\in V}\sqrt{2\,S_{\mathbb{M}}(\Psi^{i-1}(\rho)\|\Psi_v(\Psi^{i-1}(\rho)))}\nonumber\\
   &  \le {N}\sqrt{2\sum_{i=1}^n\sum_{v\in V}{S_{\mathbb{M}}(\Psi^{i-1}(\rho)\|\Psi_v(\Psi^{i-1}(\rho)))}}\nonumber\\
      &\le N\sqrt{2\sum_{i=1}^n\sum_{v\in V}\left(S(\Psi^{i-1}(\rho)\|\omega_\beta)-S(\tr_v\Psi^{i-1}(\rho)\|\tr_v\omega_\beta)\right)}\nonumber\\
   &\le N\sqrt{2\sum_{i=1}^n\sum_{v\in V}\left(S(\Psi^{i-1}(\rho)\|\omega_\beta)-S(\Psi_v(\Psi^{i-1}(\rho))\|\omega_\beta)\right)}\nonumber\\
   &\overset{(1)}{\le} N \sqrt{2n\sum_{i=1}^n\left(S(\Psi^{i-1}(\rho)\|\omega_\beta)-S(\Psi(\Psi^{i-1}(\rho))\|\omega_\beta)\right)}\nonumber\\
   &\le N\sqrt{2n\,S(\rho\|\omega_\beta)}\,.\label{eq:2step2}
\end{align}
Inequality (1) above uses the concavity of the entropy, so that for any state $\rho$
\begin{align}
    \frac{1}{n}\sum_{v\in V}\left(S(\rho\|\omega_\beta)-S(\Psi_v(\rho)\|\omega_\beta)\right)&=S(\rho\|\omega_\beta)+\frac{1}{n} \sum_{v\in V} S(\Psi_v(\rho))+\frac{1}{n}\sum_{v\in V}\tr\left[\Psi_v(\rho)\ln\omega_\beta\right]\nonumber\\
    &\le S(\rho\|\omega_\beta)+ S(\Psi(\rho))+\tr\left[\Psi(\rho)\ln\omega_\beta\right]\nonumber\\
    &=S(\rho\|\omega_\beta)-S(\Psi(\rho)\|\omega_\beta)\,.
\end{align}
Plugging \eqref{eq:1step2} and \eqref{eq:2step2} onto \eqref{equationsum2}, the result follows.
\end{proof}

It remains to prove that \eqref{contraction2} is satisfied at high enough temperature.

\begin{prop}\label{prop:curvaturebetac}
There exists an inverse temperature $\beta_c>0$ such that for all $\beta<\beta_c$, \eqref{contraction2} holds for some constant $\kappa(\beta)>0$. 
In particular, whenever $N>1$, one can choose 
\begin{align}\label{boundbetac}
\beta_c=(5N\max_{A\in E}\|h_A\|_\infty)^{-1}W\Big(\frac{1}{16d^3}\Big)\,,
\end{align}
where $W$ denotes the Lambert function and is defined as the inverse of $x\mapsto xe^x$.
\end{prop}

\begin{proof}
We have
\begin{equation}
    \left\|\Psi\right\|_{W_1\to W_1} =  \max_{\Delta\in\mathcal{O}_V^T}\frac{\left\|\Psi(\Delta)\right\|_{W_1}}{\left\|\Delta\right\|_{W_1}}\,.
\end{equation}
Any $\Delta\in\mathcal{O}_V^T$ can be expressed as \cite[Section III]{de2020quantum}
\begin{equation}
    \Delta = \sum_{v\in V}\Delta_v
\end{equation}
such that for any $v\in V$, $\Delta_v\in\mathcal{O}_V^T$ satisfies $\mathrm{Tr}_v\Delta_v=0$ and
\begin{equation}
    \left\|\Delta\right\|_{W_1} = \sum_{v\in V}\left\|\Delta_v\right\|_{W_1} = \frac{1}{2}\sum_{v\in V}\left\|\Delta_v\right\|_1\,.
\end{equation}
Therefore, we have
\begin{equation}
   \left\|\Psi\right\|_{W_1\to W_1} = \max_{v\in V}\max\left\{\left\|\Psi(\Delta_v)\right\|_{W_1}:\Delta_v\in\mathcal{O}_V^T,\,\mathrm{Tr}_v\Delta_v=0,\,\left\|\Delta_v\right\|_1=2\right\}\,.
\end{equation}
We have
\begin{align}\label{eq:ineq1}
    \left\|\Psi(\Delta_v)\right\|_{W_1} &\le \left\|\Psi(\Delta_v) - \frac{\mathbb{I}_v}{d}\otimes\mathrm{Tr}_v\Psi(\Delta_v)\right\|_{W_1} + \left\|\frac{\mathbb{I}_v}{d}\otimes\mathrm{Tr}_v\Psi(\Delta_v)\right\|_{W_1} \nonumber\\
    &= \frac{1}{2}\left\|\Psi(\Delta_v) - \frac{\mathbb{I}_v}{d}\otimes\mathrm{Tr}_v\Psi(\Delta_v)\right\|_1 + \left\|\mathrm{Tr}_v\Psi(\Delta_v)\right\|_{W_1}\nonumber\\
    &\le \frac{1}{2}\left\|\Psi(\Delta_v)\right\|_1 + \frac{1}{2}\left\|\mathrm{Tr}_v\Psi(\Delta_v)\right\|_1 + \left\|\mathrm{Tr}_v\Psi(\Delta_v)\right\|_{W_1}\,,
\end{align}
where the equality follows from \autoref{prop2} and \autoref{prop4}.
Since $\mathrm{Tr}_v\Delta_v=0$, we have $\Psi_v(\Delta_v) = 0$, and
\begin{equation}\label{eq:ineq2}
  \frac{1}{2}\left\|\Psi(\Delta_v)\right\|_1 \le \frac{1}{2n}\sum_{w\in V\setminus v}\left\|\Psi_w(\Delta_v)\right\|_1 \le 1 - \frac{1}{n}\,.  
\end{equation}
For any $w\in V\setminus N_v$ we have
\begin{equation}
    \mathrm{Tr}_v\Psi_w(\Delta_v) = \Psi_w(\mathrm{Tr}_v\Delta_v) = 0\,.
\end{equation}
Then,
\begin{equation}
    \mathrm{Tr}_v\Psi(\Delta_v) = \frac{1}{n}\sum_{w\in N_v\setminus v}\mathrm{Tr}_v\Psi_w(\Delta_v)\,.
\end{equation}
We have for any $w\in N_v\setminus v$, recalling that $v\in N_w$,
\begin{equation}
    \mathrm{Tr}_{N_w\setminus v}\mathrm{Tr}_v\Psi_w(\Delta_v) = \mathrm{Tr}_{N_w}\Psi_w(\Delta_v) = \mathrm{Tr}_{N_w}\Delta_v = 0\,,
\end{equation}
therefore,
\begin{equation}
    \left\|\mathrm{Tr}_v\Psi_w(\Delta_v)\right\|_{W_1} \le \left(N-1\right)\left\|\mathrm{Tr}_v\Psi_w(\Delta_v)\right\|_1\,,
\end{equation}
and
\begin{equation}\label{eq:ineq3}
    \left\|\mathrm{Tr}_v\Psi(\Delta_v)\right\|_{W_1} \le \frac{N-1}{n}\sum_{w\in N_v\setminus v}\left\|\mathrm{Tr}_v\Psi_w(\Delta_v)\right\|_1\,.
\end{equation}
Moreover,
\begin{equation}\label{eq:ineq4}
    \left\|\mathrm{Tr}_v\Psi(\Delta_v)\right\|_1 \le \frac{1}{n}\sum_{w\in N_v\setminus v}\left\|\mathrm{Tr}_v\Psi_w(\Delta_v)\right\|_1\,.
\end{equation}
Putting together \eqref{eq:ineq1}, \eqref{eq:ineq2}, \eqref{eq:ineq3} and \eqref{eq:ineq4}, we get
\begin{align}\label{eq:sum}
    \left\|\Psi(\Delta_v)\right\|_{W_1} &\le 1 - \frac{1}{n} + \frac{N-\frac{1}{2}}{n}\sum_{w\in N_v\setminus v}\left\|\mathrm{Tr}_v\Psi_w(\Delta_v)\right\|_1\nonumber\\
    &\le 1 - \frac{1}{n} + \frac{N-\frac{1}{2}}{n}\sum_{w\in N_v\setminus v}\left(\left\|\omega_w\otimes\mathrm{Tr}_{vw}\Delta_v\right\|_1 + 2\left\|\Psi_w - \omega_w\otimes\mathrm{Tr}_w\right\|_\diamond\right)\nonumber\\
    &= 1 - \frac{1}{n} + \frac{2N-1}{n}\sum_{w\in N_v\setminus v}\left\|\Psi_w - \omega_w\otimes\mathrm{Tr}_w\right\|_\diamond\,,
\end{align}
where $\omega_w\otimes\mathrm{Tr}_w$ is the quantum channel that replaces with $\omega_w$ the state of the site $w$.
We then have
\begin{equation}
\left\|\Psi\right\|_{W_1\to W_1} \le 1 - \frac{1}{n} + \frac{2N-1}{n}\sum_{w\in N_v\setminus v}\left\|\Psi_w - \omega_w\otimes\mathrm{Tr}_w\right\|_\diamond\,.
\end{equation}
We have
\begin{equation}\label{eq:intt}
    \left\|\Psi_w - \omega_w\otimes\mathrm{Tr}_w\right\|_\diamond \le \int_{\mathbb{R}}\left\|\omega^\frac{1-it}{2}\,\omega_{w^c}^\frac{it-1}{2}\left(\mathbb{I}_w\otimes\mathrm{Tr}_w\left[\cdot\right]\right)\omega_{w^c}^{-\frac{it+1}{2}}\,\omega^\frac{1+it}{2} - \omega_w\otimes\mathrm{Tr}_w\left[\cdot\right]\right\|_\diamond d\mu_0(t)\,.
\end{equation}
Since the Hamiltonian terms $h_A$ commute we have that, given $H_v:=\sum_{A\ni v}h_A$,
\begin{align}
    \omega^\frac{1-it}{2}\,\omega_{v^c}^\frac{it-1}{2}=e^{-\beta\frac{1-it}{2}H_v}\Big(\tr_v\big[e^{-\beta H_v}\big]\Big)^{\frac{it-1}{2}}\,.
\end{align}
Now,
\begin{align}
   \Big\| \omega^{\frac{1-it}{2}}\omega_{v^c}^{\frac{it-1}{2}}-d^{\frac{it-1}{2}}\mathbb{I}\Big\|_\infty &\le \Big\|e^{-\beta\frac{1-it}{2}H_v}-\mathbb{I}\Big\|_\infty\,\Big\|\Big(\tr_v\big[e^{-\beta H_v}\big]\Big)^{\frac{it-1}{2}}\Big\|_\infty\nonumber\\
   &\quad + \,\Big\|\Big(\tr_v\big[e^{-\beta H_v}\big]\Big)^{\frac{it-1}{2}}-d^{\frac{it-1}{2}}\mathbb{I}\Big\|_\infty\nonumber\\
   &\overset{(1)}{\le} \beta\,\frac{\sqrt{1+t^2}}{2}\|H_v\|_\infty\,e^{\beta\frac{\sqrt{1+t^2}}{2}\|H_v\|_\infty }d^{-\frac{1}{2}}\,e^{\frac{\beta}{2}\|H_v\|_\infty}\nonumber\\
   &\quad +\frac{\sqrt{1+t^2}}{2}\,M^{1+\frac{\sqrt{1+t^2}}{2}}\,d\Big\|d^{-1}\tr_v\big[e^{-\beta H_v}\big]-\mathbb{I}\Big\|_\infty\nonumber\\
   &\le {\beta\,\sqrt{1+t^2}\|H_v\|_\infty}\,d^{1+\frac{\sqrt{1+t^2}}{2}}e^{\beta\|H_v\|_\infty\big(2+\frac{\sqrt{1+t^2}}{2}\big)}\nonumber\\
   &\equiv f_v(\beta,t)\,.
\end{align}
 Inequality (1) above follows from the operator convexity of $x\mapsto x^{-\frac{1}{2}}$ as well as \autoref{technlemm}, where $M:=\max\{\|\tr_v\big[e^{-\beta H_v}\big]\|_\infty,\|\tr_v\big[e^{-\beta H_v}\big]^{-1}\|_\infty,d\}\le d\,e^{\beta\|H_v\|_\infty}$. Moreover,
\begin{align}
 e^{-2\beta\|H_v\|_\infty}d^{-1}\mathbb{I}\le    \omega_v\le e^{2\beta\|H_v\|_\infty}d^{-1}\mathbb{I}\quad \Rightarrow\quad \|\omega_v-d^{-1}\mathbb{I}\|_1  \le 2\beta\|H_v\|_\infty\,e^{2\beta\|H_v\|_\infty}\,.
\end{align}
Therefore,
\begin{align}
   & \left\|\omega^\frac{1-it}{2}\,\omega_{v^c}^\frac{it-1}{2}\left(\mathbb{I}_v\otimes\mathrm{Tr}_v\left[\cdot\right]\right)\omega_{v^c}^{-\frac{it+1}{2}}\,\omega^\frac{1+it}{2} - \omega_v\otimes\mathrm{Tr}_v\left[\cdot\right]\right\|_\diamond\nonumber \\
    &\qquad\qquad  \le  d^{\frac{1}{2}}\big(e^{\beta\|H_v\|_\infty}+1\big)\,f_v(\beta,t)+\|d^{-1}\mathbb{I}-\omega_v\|_\infty\nonumber\\
    &\qquad\qquad \le d^{\frac{1}{2}}\big(e^{\beta\|H_v\|_\infty}+1\big)\,f_v(\beta,t)+ 2\beta\|H_v\|_\infty\,e^{2\beta\|H_v\|_\infty}\,,
\end{align}
and the integrand in \eqref{eq:intt} tends to zero pointwise for $\beta\to0$.
On the other hand, we have for any $t\in\mathbb{R}$
\begin{align}
    &\left\|\omega^\frac{1-it}{2}\,\omega_{v^c}^\frac{it-1}{2}\left(\mathbb{I}_v\otimes\mathrm{Tr}_v\left[\cdot\right]\right)\omega_{v^c}^{-\frac{it+1}{2}}\,\omega^\frac{1+it}{2} - \omega_v\otimes\mathrm{Tr}_v\left[\cdot\right]\right\|_\diamond\nonumber\\
    &\le \left\|\omega^\frac{1-it}{2}\,\omega_{v^c}^\frac{it-1}{2}\left(\mathbb{I}_v\otimes\mathrm{Tr}_v\left[\cdot\right]\right)\omega_{v^c}^{-\frac{it+1}{2}}\,\omega^\frac{1+it}{2}\right\|_\diamond + \left\|\omega_v\otimes\mathrm{Tr}_v\left[\cdot\right]\right\|_\diamond\nonumber\\
    &\le 2\,,
    \end{align}
therefore the integrand in \eqref{eq:intt} is uniformly bounded.
Then, we get for all $t\in\mathbb{R}_+$ that
\begin{align}\label{bound}
    \left\|\Psi_v - \omega_v\otimes\mathrm{Tr}_v\right\|_\diamond \le d^{\frac{1}{2}}\big(e^{\beta\|H_v\|_\infty}+1\big)\,f_v(\beta,t)+ 2\beta\|H_v\|_\infty\,e^{2\beta\|H_v\|_\infty} +2\mu_0([-t,t]^c)
\end{align}
 Therefore, for any $0<\kappa<1$ there exists $\beta(\kappa)>0$ such that condition \eqref{contraction2} is satisfied for all $0 \le \beta \le \beta(\kappa)$. More precisely, in view of 
 \eqref{bound} and \eqref{eq:sum}, it is sufficient that 
 \begin{align}\label{eq:ODI}
  4 \, {\beta\,\sqrt{1+t^2} C}\,d^{\frac{3+\sqrt{1+t^2}}{2}}e^{\beta C\big(3+\frac{\sqrt{1+t^2}}{2}\big)} +2\mu_0([-t,t]^c)\le \frac{N-1}{2N-1}
 \end{align}
where $C:=\sup_v\|H_v\|_\infty$. Moreover, it is clear that $\mu_0([-t,t]^c)\le 2e^{-\pi t}$. The result follows after choosing $t$ so that the exponentially decaying term $4e^{-\pi t}$ counts for at most half the upper bound and solving \eqref{eq:ODI} for $\beta_c$, up to some numerical simplifications. 
\end{proof}

\begin{rem}
The lower bound \eqref{boundbetac} can be compared to that in the classical setting \cite[Example 17]{Ollivier2009} (see also \cite{Griffiths1967}): there, the author showed that for a Hamiltonian of the form $U(S):=-\sum_{x\sim y\in G}S(x)S(y)-H\sum_xS(x)$, where $S(x)\in\{-1,1\}$ denotes the spin configuration at the site $x$ of a graph $G$, i.e. $d=2$, 
\begin{align*}
    \beta_c\ge \frac{1}{2}\,\ln\Big(\frac{N+1}{N-1}\Big)\sim_{N\to\infty}\frac{1}{N}\,,
\end{align*}
which shows asymptotic optimality of our result, up to numerical multiplicative constants. For comparison, the exact value of $\beta_c$ for the Ising model on the regular infinite tree with degree $N$ is known to be equal to $\frac{1}{2}\ln\big(\frac{N}{N-2}\big)$. 
\end{rem}

\subsection{Auxiliary lemma}
\begin{lem}\label{technlemm}
For any positive, definite matrices $A,B$ and all $z\in\mathbb{C}$,
\begin{align}
    \|A^z-B^z\|_\infty\le |z|\,\max\{\|A\|_\infty,\|A^{-1}\|_\infty,\|B\|_\infty,\|B^{-1}\|_\infty\}^{1+|\operatorname{Re}(z)|}\,\|A-B\|_\infty\,,
\end{align}
\end{lem}
\begin{proof}
It suffices to use a linear interpolation between $A$ and $B$: $A(s):=sA+(1-s)B$. We have
\begin{align}
    A^{z}-B^z&=\int_{0}^1\,\frac{d}{ds}A(s)^{z}\,ds\nonumber\\
    &=z\,\iint_{[0,1]^2}\,A(s)^{zu}\,\frac{d}{ds}\ln(A(s))\,A^{z(1-u)}\,dsdu\nonumber\\
    &=z\iint_{[0,1]^2}\int_0^\infty\,A(s)^{zu}(A(s)+v)^{-1}\,(A-B)(A(s)+v)^{-1}A(s)^{z(1-u)}\,dvduds\,.
\end{align}
Then,
\begin{align}
     \|A^{z}-B^z\|_\infty&\le |z|\,\int_{0}^1\int_0^\infty\,\|A(s)^{\operatorname{Re}(z)}\|_\infty\,\|(A(s)+v)^{-1}\|_\infty^2\,\|A-B\|_\infty\,dvds\nonumber\\
     &\le |z|\,\|A-B\|_\infty\,M(z)\,\int_0^1\|A(s)^{-1}\|_\infty\,ds\nonumber\\
     &\le |z|\,M(z)\cdot M'\,\|A-B\|_\infty\,.
\end{align}
by the operator convexity of $x\mapsto x^{-1}$ where $M(z):=\max_{s\in[0,1]}\|(sA+(1-s)B)^{\operatorname{Re}(z)}\|_\infty$ and $M':=\max\{\|A^{-1}\|_\infty,\|B^{-1}\|_\infty\}$. The result follows by operator convexity of the inverse function and further simple estimates.
\end{proof}

\section{Modified logarithmic Sobolev inequalities}\label{MLSI}

In this section, we pursue a different approach to prove transportation cost inequalities for $W_1$, namely through the existence of a non-commutative entropic inequality known as the \textit{modified logarithmic Sobolev inequality} \cite{Kastoryano2013,gao2021spectral}. In order to introduce our main result, we need a variation of the Lipschitz constant that was introduced in \cite{Rouz2019}. This definition departs from a noncommutative differential structure, which we define below (see \cite{Carlen_2019}):

\begin{defn}[Differential structure]\label{thm:normalformCM}
A set of operators $L_k \in \mathcal{O}_V$ and constants $\omega_k\in\mathbb{R}$ define a differential structure $\{L_k,\omega_k\}_{k\in\mathcal{K}}$ for a full rank state $\omega\in\mathcal{S}_V$ if
					\begin{itemize}
						\item[1] $\{L_k\}_{k\in\mathcal{K}}=\{L_k^{\dagger}\}_{k\in\mathcal{K}}$;
						\item[2] $\{L_k\}_{k\in\mathcal{K}}$ consists of eigenvectors of the modular operator $\Delta_\omega(X):=\omega X\omega^{-1}$ with
						\begin{align}\label{eigenD}
							\Delta_\omega(L_k)=e^{-\omega_k}L_k\,.
						\end{align}
						\end{itemize}

\end{defn}
Such a differential structure can be used to provide the set of matrices with a Lipschitz constant that is tailored to $\omega$, see e.g.~\cite{Rouz2019,Carlen_2019} for more on this. In order to distinguish that constant from $\|.\|_L$, we refer to it as the differential Lipschitz constant and denote it by $\vertiii{\nabla X}$.
It is defined as:
\begin{align}\label{equ:lipnorm}
\vertiii{\nabla X}:= \left(  \sum_{k\in\mathcal{K}}  (e^{-\omega_k/2}+e^{\omega_k/2})\|\partial_kX\|_{\infty}^2\right)^{1/2}\,,
\end{align}
where $\partial_k X\equiv [L_k,X]$. For ease of notations, we will denote the differential structure by the couple $(\nabla,\omega)$. The notion of a differential structure is also intimately connected to that of the generator of a quantum dynamical semigroup converging to $\omega$~\cite{Carlen_2019}, and properties of that semigroup immediately translate to properties of the metric. This is because the differential structure can be used to define an operator that behaves in an analogous way to the Laplacian on a smoth manifold, which in turn induces a heat semigroup. We refer to~\cite{Carlen_2019,Rouz2019} for more details on this connection and interpretation.

When the state $\omega$ is a quantum Gibbs state corresponding to a local, commuting Hamiltonian associated to a uniformly bounded interaction defined on a lattice $V\subset\subset \mathbb{Z}^D$, the differential structure $(\nabla,\omega)$ can be chosen as local. This means that the operators $L_k\equiv L_{i,\alpha}$ are indexed by a site $i\in V $ and an index $\alpha$ of a set $\Gamma$ whose cardinality only depends on the local dimension $d$ and the locality $\kappa$ of $\omega$. Moreover, we assume that the operators $L_{i,\alpha}$ are supported on a neighborhood $\mathcal{N}_i$ of site $i$ of diameter $r\equiv r(\kappa)$ and the corresponding constants $\omega_{i,\alpha}$ are uniformly bounded: $\sup_{i\in\mathbb{Z}^D}\max_{\alpha \in\Gamma}|\omega_{i,\alpha}|\equiv \Omega<\infty$. The definition in Eq.~\eqref{equ:lipnorm} yields a metric on states by duality:
\begin{align*}
    W_{\nabla}(\rho,\omega):=\sup\limits_{X=X^\dagger,\, \vertiii{\nabla X}\leq1}\left|\operatorname{Tr}\left(X(\rho-\omega)\right)\right|.
\end{align*} 

\begin{prop}\label{prop:comparison}
Given the Gibbs state $\omega$ of a local commuting Hamiltonian $H$ on $V\subset\subset \mathbb{Z}^D$ with $|V|=n$ and associated local differential structure $(\nabla,\omega)$, the following bound holds for all $\rho\in \mathcal{S}_V$: 
\begin{align*}
    \left\|\rho-\omega\right\|_{W_1}\le C\,\sqrt{n}\,W_{\nabla}(\rho,\omega)\,,
\end{align*}
for some constant $C$ independent of $n$.
\end{prop}

\begin{proof}
By duality, it is equivalent to prove that for all $H\in\mathcal{O}_V$
\begin{align*}
    \vertiii{\nabla H}\le C\,\sqrt{n}\,\|H\|_L\,.
\end{align*}
First, we have
\begin{align}
    \vertiii{\nabla H}&=\Big(\sum_{i\in V}\sum_{\alpha\in\Gamma}(e^{-\omega_{i,\alpha}/2}+e^{\omega_{i,\alpha}/2})\,\|[L_{i,\alpha},H]\|_\infty^2\Big)^{\frac{1}{2}}\\
    &\le \sqrt{n|\Gamma|}\,\sqrt{2\,e^{\Omega/2}}\,\max_{i\in V}\,\max_{\alpha\in\Gamma}\,\|[L_{i,\alpha},H]\|_\infty\,.
\end{align}
Now, since for each pair $(i,\alpha)$, $L_{i,\alpha} $ is supported on a neighborhood $\mathcal{N}_i$ of site $i\in V$, 
\begin{align}
\|[L_{i,\alpha},H]\|_\infty&= \big\|[L_{i,\alpha},H-\frac{\mathbb{I}_{\mathcal{N}_i}}{d^{|\mathcal{N}_i|}}\otimes\mathrm{Tr}_{\mathcal{N}_i}(H)]\big\|_\infty\\
&\le 2\,\|L_{i,\alpha}\|_\infty\,\big\|H-\frac{\mathbb{I}_{\mathcal{N}_i}}{d^{|\mathcal{N}_i|}}\otimes\mathrm{Tr}_{\mathcal{N}_i}(H)\big\|_\infty\,.
\end{align}
Next, by a telescopic sum argument, we can further control the last infinity norm on the right hand side above as follows: given an arbitrary ordering of the region $\mathcal{N}_i$,
\begin{align}
    \big\|H-\frac{\mathbb{I}_{\mathcal{N}_i}}{d^{|\mathcal{N}_i|}}\otimes\mathrm{Tr}_{\mathcal{N}_i}(H)\big\|_\infty&\overset{(1)}{\le}
    \sum_{j=1}^{|\mathcal{N}_i|}
    \Big\|\Big(\frac{\mathbb{I}_{1\ldots j-1}}{d^{j-1}}\otimes\mathrm{Tr}_{1\ldots j-1}-
     \frac{\mathbb{I}_{1\ldots j}}{d^{j}}\otimes\mathrm{Tr}_{1\ldots j}\Big)(H)\Big\|_\infty\\
     &\overset{(2)}{\le} |\mathcal{N}_i| \max_{j\in\mathcal{N}_i}\,\|H-\frac{\mathbb{I}_j}{d}\otimes \mathrm{Tr}_{j}(H)\|_\infty\,,
\end{align}
where $(1)$ follows from the triangle inequality whereas $(2)$ follows from the fact that the maps $\frac{\mathbb{I}_{1\ldots j-1}}{d^{j-1}}\otimes\mathrm{Tr}_{1\ldots j-1}$ are completely positive and unital, and therefore contract the operator norm. All in all, we have derived the following bound on the differential Lipschitz constant of $H$:
\begin{align}
    \vertiii{\nabla{H}}&\le \sqrt{n|\Gamma|}\,2\sqrt{2\,e^{\Omega/2}}\,\max_{i\in V}\,\max_{\alpha\in\Gamma}\,\|L_{i,\alpha}\|_\infty\, |\mathcal{N}_i|\,\max_{j\in \mathcal{N}_i}\,\|H-\frac{\mathbb{I}_j}{d}\otimes \mathrm{Tr}_{j}(H)\|_\infty\\
    &\overset{(3)}{\le} \frac{d^2-1}{d^2}\sqrt{n|\Gamma|}\,2\sqrt{2\,e^{\Omega/2}}\,\max_{i\in V}\,\max_{\alpha\in\Gamma}\,\|L_{i,\alpha}\|_\infty\, |\mathcal{N}_i|\,\|H\|_L\\
    &\equiv C\,\sqrt{n}\,\max_{i\in V}\|H-\frac{\mathbb{I}_j}{d}\otimes \mathrm{Tr}_{j}(H)\|_\infty\,,
\end{align}
for some constant $C$ independent of $n$, and where $(3)$ follows from \autoref{prop15}. 
\end{proof}

The advantage of $W_1$ as compared to $W_{\nabla}$ is that it does not depend on the state $\omega$. On the other hand, the bound derived in \autoref{prop:comparison} can be used in conjunction with recently proved transportation cost inequalities for $W_{\nabla}$ through the proof of the existence of a modified logarithmic Sobolev inequality in order to get analogous inequalities for $W_1$ (see \cite{capel2020modified} for more details):
\begin{thm}\label{thm:LS}
Let $\omega$ be the Gibbs state of a local commuting Hamiltonian $H$ at inverse temperature $\beta$ on $V\subset\subset \mathbb{Z}^D$. Then, there exists a critical inverse temperature $\beta_c$ such that $C(\omega_\beta)\le C\,n$ for some constant $C$ independent of $n=|V|$ whenever $\beta<\beta_c$ if any of the two conditions below is satisfied:
\begin{itemize}
\item[$\operatorname{(i)}$] $H$ is classical;
\item[$\operatorname{(ii)}$] $H$ is a nearest neighbour Hamiltonian.
\end{itemize}
Moreover, we can drop the assumption of $2$-locality in the $1$D case, where $\beta_c=0$ at the cost of getting a slightly worsened constant $C(\omega_\beta)\le Cn\operatorname{polylog}(n)$, so that we recover the result of \autoref{thm:Markov}.
\end{thm}

\begin{proof}
In \cite{capel2020modified,bardet2021entropy,bardet2021rapid} the existence of local differential structures associated to $\omega$ that satisfy the so-called modified logarithmic Sobolev inequality was proved under the conditions of the theorem. Moreover, the modified logarithmic Sobolev inequality implies the transportation cost inequality for the differential Wasserstein distance \cite{Rouz2019}: there exists a constant $C'$ independent of $n$ such that
\begin{align}
    W_{\nabla}(\rho,\omega)\le \sqrt{C'\,S(\rho\|\omega)}
\end{align}
for all state $\rho\in\mathcal{S}_V$. This fact in conjunction with \autoref{prop:comparison} allows us to conclude.
\end{proof}

\section{Local indistinguishability}\label{localindistinguishability}

In this section, we provide a transportation cost inequality under a condition of local indistinguishability \cite{cubitt2015stability,kastoryano2016quantum,brandao2019finite}. In the classical setting, this condition constitues a weakening of Dobrushin Shlosman's mixing condition \cite{dobrushin1987completely} recently considered by Marton \cite{marton2019logarithmic}. Moreover, as opposed to the latter, our technique has the benefit of not requiring the local specifications of the state to be uniformly lower bounded by a positive number, at the cost of getting a slightly worsened constant.

\subsection{Transportation cost from local indistinguishability}\label{sec:localindistinguis}
We start by proving our general result in the quantum setting. Here, we assume that the $n=(2m+1)^D$ qudits are arranged on a $D$-dimensional regular lattice $V:=[-m,m]^D$. Before we state our main result, we need to introduce the notion of a non-commutative conditional expectation.

\begin{defn}[Conditional expectations]
Let $\cN\subseteq  \cB_V$ be a von Neumann subalgebra\footnote{We recall that a finite dimensional von Neumann algebra is a matrix algebra that is close under taking the adjoint.} of $\cB_V$. A conditional expectation onto $\cN$ is a completely positive unital map $E_\cN^\dagger:\cB_V\to \cN$ satisfying \begin{enumerate}
\item[i)]for all $X\in \cN$, $E^\dagger_\cN(X)=X$;
\item[ii)]for all  $a,b\in\cN,X\in \cB_V$, $E^\dagger_\cN(aXb)=aE_\cN(X)b$. 
 \end{enumerate}
We denote by $E_{\cN}$ its adjoint map with respect to the trace inner product, i.e.
\begin{align*}
\tr(E_{\cN}(X)Y)=\tr(XE_\cN^\dagger(Y))\,.
\end{align*}
\end{defn}
As a simple example, we consider a full-rank state $\sigma \in\mathcal{S}_V$ and let $(e^{t\mathcal{L}})_{t\ge 0}$ be a quantum Markov semigroup. Under the following detailed balance condition, the limit $\lim_{t\to\infty}e^{t\mathcal{L}^\dagger}=E^\dagger_{\mathcal{N}}$ is a conditional expectation onto the algebra $\mathcal{N}$ of fixed points of the semigroup:
\begin{align*}
    \forall X,Y\in\mathcal{B}_V,\quad \tr\big( \sigma\,X^\dagger \mathcal{L}^\dagger(Y)\big)=\tr\big( \sigma\,\mathcal{L}^\dagger(X)^\dagger Y\big)\,.
\end{align*}
Next, for a state $\rho$, the relative entropy with respect to $\cN$ is defined as follows
\[S(\rho\|\mathcal{N}):=S(\rho\|E_{\cN}(\rho))=\inf_{E_{\cN}(\sigma)=\sigma} S(\rho\|\sigma)\,,\]
where the infimum is always attained by $\sigma = E_{\cN}(\rho)$. Indeed, for any $\sigma$ satisfying $E_{\cN}(\sigma)=\sigma$,  we have the following chain rule (see \cite[Lemma 3.4]{junge2019stability})
\begin{align}\label{eq.chainrule}
S(\rho\|\sigma)=S(\rho\|E_{\cN}(\rho))+S(E_{\cN}(\rho)\|\sigma)\,.
\end{align}
Hence the infimum is attained if and only if $S(E_{\cN}(\rho)\|\sigma)=0$.

\begin{defn}[Local indistinguishability]\label{def:indistingui}
    Let $\{\cN_C\}_{C\subseteq V}$ be a set of subalgebras of $\cB_V$ such that $\cB_{{C}^c}\subset \cN_C$ and $\mathbf{E}:=\{E_C\}_{C\subseteq V}$ be a set of compatible conditional expectations $E^\dagger_C:\cB_V\to \cN_C$ acting non-trivially on region $C$, \emph{i.e.}, they satisfy the property that for any $C\subseteq C'$, $E_C\circ E_{C'}=E_{C'}\circ E_C= E_{C'}$. Then, we say that $\mathbf{E}$ satisfies local indistinguishability if there exists a fast decaying function $\varphi:\mathbb{N}\to\mathbb{R}$ independent of $V$ such that for every regions $XYZ\subset V$ with $\operatorname{dist}(X,Z)\ge \ell$, and for all states $\rho\in\mathcal{S}_{V}$,
    \begin{align*}
  \|E_{YZ}\circ( E_{XYZ}- E_{XY})(\rho)\|_1\le \,|XYZ|\,\varphi(\ell) \,,
    \end{align*}
\end{defn}
For instance, take a product state $\omega\in\mathcal{S}_V$ and for each region $C\subseteq V$, denote $E_C(\rho)=\tr_C(\rho)\otimes \omega_C$. One can easily verify that the maps $E_C$ are conditional expectations and satisfy the local indistinguishability condition with $\varphi=0$.
We are now ready to state and prove the main theorem of this section. For a strictly decreasing function $\varphi:\mathbb{N}\to\mathbb{R}_+$ and a positive real number $a>0$, we denote by 
 $\varphi^{-1}(a):=\min\{\ell\in\mathbb{N}:\varphi(\ell)\le a\}$. 

\begin{thm}\label{thm:localind}
Let $\mathbf{E}$ be a set of compatible conditional expectations satisfying local indistinguishability with fast decaying function $\varphi$. Then for all hypercubes $V_0\subset V$ and all $\rho\in\mathcal{S}_V$,
\begin{align}
\left\|\rho-E_{V_0}(\rho)\right\|_{W_1}\le2\sqrt{20}\,(7\varphi^{-1}(|V_0|^{-3/2}))^D\sqrt{|V_0|\,S(\rho\|E_{V_0}(\rho))}\,,
\end{align}
for some fixed constant $c$ of order $1$. In particular, whenever $E_{V}(\rho)=\omega_V\in \mathcal{S}_V$ for all states $\rho$, and assuming the exponential clustering function $\varphi(\ell):=\kappa e^{-\ell/\xi}$, the state $\omega_{V}$ satisfies $\operatorname{TC}(c)$ with $c=O({n}\operatorname{polylog}(n))$.
\end{thm}
 \begin{figure}[h!]
	\centering
	\includegraphics[width=0.4\linewidth]{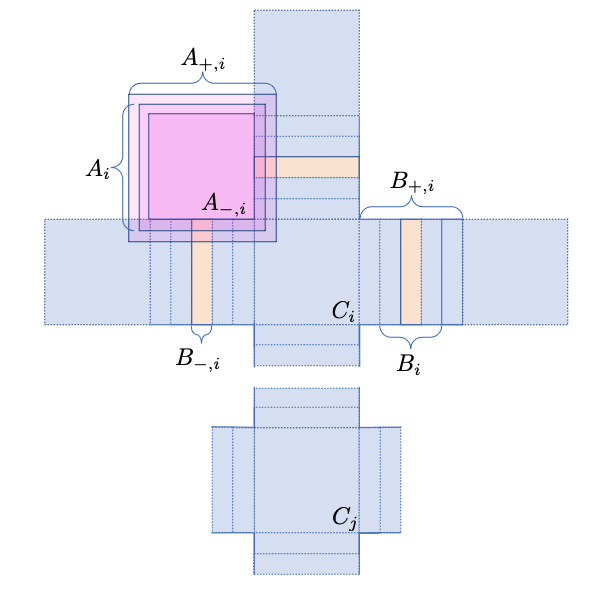}
	\caption{Geometry of the lattice in the proof of \autoref{thm:localind} with tiling by regions $A_+:=\cup_i A_{+,i}$, $B_+:=\cup_i B_{+,i}$ and $C:=\cup C_i$. }
	\label{figuretiling1}
\end{figure}

\begin{proof}

For sake of clarity, we provide the proof for $D=2$ only, although the general case follows similarly. First, we partition the hypercube $V_0$ into regions $A$, $B_+$ and $C_+$ in the same way as done in \cite{brandao2019finite} (see also \autoref{figuretiling1}). Then, by triangle inequality
\begin{align}
    &\|\rho-E_{V_0}(\rho)\|_{W_1}\nonumber\\
    &~~\le \|\rho-E_{A_+}E_{B_-C}(\rho)\|_{W_1}+\|E_{A_+}E_{B_-C}(\rho)-E_{V_0}(\rho)\|_{W_1}\,\nonumber\\
    &~~\le \|(\id-E_{A+}E_{B_+}E_{C})(\rho)\|_{W_1}+\|E_{A+}(E_{B_+}E_{C}-E_{B_-C})(\rho)\|_{W_1}+\|(E_{A_+}E_{B_-C}-E_{V_0})(\rho)\|_{W_1}\, \nonumber\\
    &~~\le (I)+ 2|A_+^{\max}|\,(II)+(III)\,\label{eq:split},
    \end{align}
    where 
    \begin{align*}
      &(I):=  \|(\id-E_{A+}E_{B_+}E_{C})(\rho)\|_{W_1}\\
      &(II):=\|(E_{B_+}E_{C}-E_{B_-C})(\rho)\|_{W_1}\\
      &(III):=\|(E_{A_+}E_{B_-C}-E_{V_0})(\rho)\|_{W_1}
    \end{align*}
 and where the last bound in \eqref{eq:split} follows from \autoref{prop13} with $|A_+^{\max}|:=\max_{i}|A_{+,i}|$. Now, we control each of the norms on the right-hand side of \eqref{eq:split} separately. First, we denote by $E_{C^{(0)}}:=\id$, $C^{(i)}:=\cup_{j\le i}C_j$ given an arbitrary ordering of the connected subregions in $C$, and similarly for the other regions $A_+$ and $B_+$. Then,
\begin{align}
    (I)&\le \|(\id-E_{C})(\rho)\|_{W_1}+\|(E_{C}-E_{B_+}E_C)(\rho)\|_{W_1}+\|(E_{B_+}E_C-E_{A_+}E_{B+}E_C)(\rho)\|_{W_1}\label{firstineq}\\
    & \le |C^{\max}|\,\sum_{i\in\mathcal{I}_C}\|(E_{C^{(i-1)}}-E_{C^{(i)}})(\rho)\|_1+|B_+^{\max}|\sum_{i\in\mathcal{I}_{B+}}\|(E_{B_+^{(i-1)}}-E_{B_+^{(i)}})(E_{C}(\rho))\|_1\nonumber\\
    &~~~+|A_+^{\max}|\sum_{i\in\mathcal{I}_{A+}}\|(E_{A_+^{(i-1)}}-E_{A_+^{(i)}})E_{B_+}E_{C}(\rho)\|_1\label{eq:prop4}\\
    &\le |V_0^{\max}|\,\Big\{\,\sqrt{2|\mathcal{I}_C|}\,\sqrt{\sum_{i\in\mathcal{I}_C}\,S(E_{C^{(i-1)}}(\rho)\|E_{C^{(i)}}(\rho))}\nonumber\\
    &~~~~~~~~~~~~~+\,\sqrt{2|\mathcal{I}_{B_+}|}\,\sqrt{\sum_{i\in\mathcal{I}_{B+}}\,S(E_{B_+^{(i-1)}}(E_C(\rho))\|E_{B_+^{(i)}}(E_C(\rho)))} \nonumber\\
    &~~~~~~~~~~~~~+\sqrt{2|\mathcal{I}_{A_+}|}\,\sqrt{\sum_{i\in\mathcal{I}_{A_+}}\,S(E_{A_+^{(i-1)}}E_{B_+}E_C(\rho)\|E_{A_+^{(i)}}E_{B_+}E_C(\rho))}\Big\}\label{eq:Pinsker}\\
    &\le |V_0^{\max}|\,\sqrt{2|\mathcal{I}_{\max}|}\,\Big\{ \sqrt{S(\rho\|E_C(\rho))}+\sqrt{S(E_{C}(\rho)\|E_{B_+}E_C(\rho))}\nonumber\\
    &~~~~~~~~~~~~~~~~~~~~~~~~~~~~~~~~~~~~~~~~~~~~~~+\sqrt{S(E_{B_+}E_C(\rho)\|E_{A_+}E_{B_+}E_C(\rho))}\Big\}\label{eqfirst}
\end{align}
with
\begin{align*}
&|C^{\max}|:=\max_{i}|C_{i}|\\  &|V_0^{\max}|:=\max\{|A_+^{\max}|,|B_+^{\max}|,|C^{\max}|\}\\
&|\mathcal{I}_{\max}|:=\max\{|\mathcal{I}_{A_+}|,  |\mathcal{I}_{B_+}| ,|\mathcal{I}_C|\}\\
&
\end{align*}
where $|\mathcal{I}_C|$ denotes the number of connected components in $C$, and similarly for the other sets.
Above, \eqref{firstineq} follows by the triangle inequality, \eqref{eq:prop4} by the triangle inequality and \autoref{prop5}, and \eqref{eq:Pinsker} by Pinsker's inequality as well as Jensen's inequality for $x\mapsto x^{\frac{1}{2}}$. \eqref{eqfirst} follows from the chain rule in \eqref{eq.chainrule}, and the fact that the regions $C_i$, resp. $A_{+,i}$, resp. $B_{+,i}$, do not overlap, so that for instance $E_{A_{+,i}}E_{A_{+,j}}=E_{A_{+,j}}E_{A_{+,i}}=E_{A_{+,j}\cup A_{+,i}}$. Next, we control the second norm on the right-hand side of \eqref{eq:split}: using \autoref{prop5} , we have
\begin{align}
    (II)&\le 2|B_+|\,\|(E_{B_+}E_{C}-E_{B_-C})(\rho)\|_{1}\nonumber\\
    &= 2|B_+|\,\|(E_{B+}E_C-E_{B_-C})(E_C-E_{B_-C})(\rho)\|_1\label{eq:compatibility}\\
    &\le 2|B_+|\,\|E_{B+}E_C-E_{B_-C}\|_{1\to 1}\,\|(E_C-E_{B_-C})(\rho)\|_1\nonumber \\
&  \le  2|B_+|\,\max_{\rho'\in\mathcal{S}_V}\,\|E_B(E_C-E_{B_-C})(\rho')\|_1\,\|E_C(\rho)-E_{B_-C}(\rho)\|_1\label{eq:dpi}\\
      &\le 2\,|B_+|\,|B_-C|\,\varphi(\ell)\,\|E_C(\rho)-E_{B_-C}(\rho)\|_1 \label{eqlocindistuse}\\
     &\le2\,|B_+|\,|B_-C|\,\varphi(\ell)\,\sqrt{2\,S(E_C(\rho)\|E_{B_-C}(\rho))}\label{eqsec}\,.
\end{align}
Above, \eqref{eq:compatibility} follows form the compatibility of the conditional expectations, and \eqref{eq:dpi} from $E_{B_+}\circ E_B=E_{B_+}$ and the monotonicity of the trace-distance under such CPTP map. \eqref{eqlocindistuse} follows from the condition of local indistinguishability
when taking $X=C\backslash B$, $Y=B\backslash B_-$ and $Z=B_-$, and assuming that $\operatorname{dist}(X,Z)\ge \ell$. Finally, \eqref{eqsec} follows from an application of Pinsker's inequality. Similarly, we find
 \begin{equation}\label{eqthird}
  (III)\le 
      2\,|A_+|\,|V_0|\,\varphi(\ell)\,\sqrt{2\,S(E_{B_-C}(\rho)\|E_{V_0}(\rho))}
 \end{equation}
 Then, by inserting \eqref{eqfirst}, \eqref{eqsec} and \eqref{eqthird} into \eqref{eq:split}, we have 
\begin{align}
   & \|\rho-E_{V_0}(\rho)\|_{W_1}\nonumber\\
    &\le 2\sqrt{2}\max\big\{ |V_0^{\max}|\,\sqrt{|\mathcal{I}_{\max}|},   |A_+|\,|V_0|\,\varphi(\ell),|B_+|\,|B_-C|\varphi(\ell)\big\}  \nonumber\\
    &\,\Big\{ \sqrt{S(\rho\|E_C(\rho))}+\sqrt{S(E_{C}(\rho)\|E_{B_+}E_C(\rho))}+\sqrt{S(E_{B_+}E_C(\rho)\|E_{A_+}E_{B_+}E_C(\rho))}\nonumber\\
    &~~~~~~~~~~~~~~~~~~~~~~+ \sqrt{S(E_{B_-C}(\rho)\|E_{V_0}(\rho))}+\sqrt{S(E_C(\rho)\|E_{B_-C}(\rho))}\Big\}\nonumber\\
    &\le2\sqrt{10}\max\big\{ |V_0^{\max}|\,\sqrt{|\mathcal{I}_{\max}|},   |A_+|\,|V_0|\,\varphi(\ell),|B_+|\,|B_-C|\varphi(\ell)\big\}\nonumber\\
    & \Big(S(\rho\|E_C(\rho))+S(E_{C}(\rho)\|E_{B_+}E_C(\rho))+S(E_{B_+}E_C(\rho)\|E_{A_+}E_{B_+}E_C(\rho))\nonumber\\
    &~~~~~~~~~~~~~~~~~~~~~~~~~~~~~~~~~~~~~~~~~~~~~~~~~ +S(E_C(\rho)\|E_{B_-C}(\rho))+S(E_{B_-C}(\rho)\|E_{V_0}(\rho))\Big)^{\frac{1}{2}}\label{eq:Jensen2}\\
   & \le 2
   \sqrt{20}\max\big\{ |V_0^{\max}|\,\sqrt{|\mathcal{I}_{\max}|},   |A_+|\,|V_0|\,\varphi(\ell),|B_+|\,|B_-C|\varphi(\ell)\big\}S(\rho\|E_{V_0}(\rho))^{\frac{1}{2}}\label{eqchainruleagain}
    \end{align}
  where \eqref{eq:Jensen2} is another directly application of Jensen's inequality for $x\mapsto x^{\frac{1}{2}}$, whereas \eqref{eqchainruleagain} follows from two uses of the chain rule \eqref{eq.chainrule} after adding the positive term $S(E_{A_+}E_{B_+}E_C(\rho)\|E_{V_0}(\rho))$ to the square root and a final use of the data processing inequality. The result then follows after choosing the length $\ell:=\varphi^{-1}(|V_0|^{-3/2})$ so that
  \begin{align*}
      \max\big\{|A_+|\,|V_0|,\,|B_+|\,|B_-C|\big\}\,\varphi(\ell)\le |V_0|^2\varphi(\ell)\le |V_0^{\max}|\,\sqrt{|\mathcal{I}_{\max}|}\,.
  \end{align*}
  With this choice, and estimating $|V_0^{\max}|\le (7\ell)^D$ the bound found in \eqref{eqchainruleagain} can be further controlled by
  \begin{align}
      \|\rho-E_{V_0}(\rho)\|_{W_1}&\le 2\sqrt{20}\,(7\ell)^{D}\,\sqrt{|V_0|\,S(\rho\|E_{V_0}(\rho))}\\
      &\le 2\sqrt{20}\,(7\varphi^{-1}(|V_0|^{-3/2}))^D\sqrt{|V_0|\,S(\rho\|E_{V_0}(\rho))}\,.
  \end{align}
    
\end{proof}

\subsection{Classical case}

In this section, we restrict our analysis to classical conditional expectations and probability measures. In this setting, it is easy to see that the property of local indistinguishability is implied by the following condition. Here, with a slight abuse of notations, we will use the same symbol for a probability measure $\mu$ on the Borel sets of $[d]^V$ and its corresponding probability mass function. 
\begin{defn}[Local indistinguishability, classical case]\label{def:indistinguiclass}
    Let $\mu$ be a probability measure on $[d]^{V}$, and $\{\mu_C\}_{C\subseteq V}$ be a set of compatible conditional probability measures $\mu_C(.|x_{{C}^c})$ acting on the sets $[d]^C$, i.e. they satisfy the property that for any $C\subseteq C'$, $\mathbb{E}_{\mu_C}\circ \mathbb{E}_{\mu_{C'}}=\mathbb{E}_{\mu_{C'}}\circ \mathbb{E}_{\mu_C}= \mathbb{E}_{\mu_{C'}}$. Then, we say that the measure $\mu$ satisfies local indistinguishability if there exists a fast decaying function $\varphi:\mathbb{N}\to\mathbb{R}_+$ such that for every regions $V'=XYZ\subset V$ such that $\operatorname{dist}(i,j)\ge \ell$ for any $i\in X$ and $j\in Z$, 
    \begin{align*}
 \max_{x_{X}\in[d]^{X}}\,\max_{x_{{V'}^c}\in [d]^{{V'}^c}}\,\, \sum_{y_{V'}} \big|\mu_{Y|X{V'}^c}(y_Y|x_Xx_{{V'}^c})\mu_{XY|ZV'^c}(y_{XY}|y_Zx_{V'^c})-\mu_{V'|{V'}^c}(y_{V'}|x_{{V'}^c})\big|\le |V'| \varphi(\ell) \,,
    \end{align*}
where $\partial Z$ denotes the boundary of $Z$.
\end{defn}
\begin{cor}
Let $\mu$ be a probability measure on $[d]^{V}$ satisfying local indistinguishability with fast decaying function $\varphi$. Then for all $\nu<\!<\mu$,
\begin{align}
W_1(\nu,\mu)\le2\sqrt{20}\,(7\varphi^{-1}(n^{-3/2}))^D\sqrt{n\,S(\nu\|\mu)}\,.
\end{align}
Equivalently, the measure $\mu$ satisfies the following sub-Gaussian tail: for any function $f$ such that $\|f\|_L\le 1$,
\begin{align}
    \mathbb{P}_\mu\Big(|f(X)-\mathbb{E}_\mu[f(X)]|>t \Big)\le 2\,\exp\left(-\frac{t^2}{80n(7\varphi^{-1}(n^{-3/2}))^{2D}}\right)\,.
\end{align}

\end{cor}

\section{Gaussian concentration}\label{concentration}

As mentioned before, the classical transportation cost inequalities for a measure $\mu$ is equivalent to the sub-Gaussian bounds on the tail probability of any Lipschitz function $f$ of a random variable $X$ drawn according to $\mu$. One way to see this is by using the variational formulation of the relative entropy in order to bound the Laplace transform of $f(X)$. In the non-commutative setting, this leads to the following characterization of the transportation cost constant $C(\omega)$:

\begin{prop}\label{prop:dual}
For any $\omega\in\mathcal{S}_V$,
\begin{equation}\label{eq:tildeC}
C(\omega) = 4\sup_{K\in\mathcal{O}_V}\frac{\ln\mathrm{Tr}\exp\left(K + \ln\omega\right) - \mathrm{Tr}\left[\omega\,K\right]}{\left\|K\right\|_L^2}\,,
\end{equation}
and the $\sup$ can be restricted to $K\in\mathcal{O}_V$ such that $\mathrm{Tr}\left[\omega\,K\right]=0$.
\begin{proof}
Let $\tilde{C}(\omega)$ be the right-hand side of \eqref{eq:tildeC}. On the one hand, let $K\in\mathcal{O}_V$ satisfy $\mathrm{Tr}\left[\omega\,K\right] = 0$, and let
\begin{equation}
\rho = \frac{\exp\left(K + \ln\omega\right)}{\mathrm{Tr}\exp\left(K + \ln\omega\right)} \in \mathcal{S}_V\,.
\end{equation}
We have
\begin{equation}
\ln\mathrm{Tr}\exp\left(K + \ln\omega\right) = \mathrm{Tr}\left[\rho\,K\right] - S(\rho\|\omega) \le \left\|\rho - \omega\right\|_{W_1}\left\|K\right\|_L - \frac{\left\|\rho - \omega\right\|_{W_1}^2}{C(\omega)} \le \frac{C(\omega)\left\|K\right\|_L^2}{4}\,,
\end{equation}
therefore $\tilde{C}(\omega)\le C(\omega)$.

On the other hand, let $\rho\in\mathcal{S}_V$, and let $K\in\mathcal{O}_V$ such that
\begin{equation}
\left\|K\right\|_L = \frac{2\left\|\rho - \omega\right\|_{W_1}}{\tilde{C}(\omega)}\,,\qquad \mathrm{Tr}\left[\omega\,K\right]=0\,,\qquad \mathrm{Tr}\left[\rho\,K\right] = \frac{2\left\|\rho - \omega\right\|_{W_1}^2}{\tilde{C}(\omega)}\,.
\end{equation}
We have
\begin{equation}
S(\rho\|\omega) \ge \mathrm{Tr}\left[\rho\,K\right] - \ln\mathrm{Tr}\exp\left(K + \ln\omega\right) \ge \frac{\left\|\rho - \omega\right\|_{W_1}^2}{\tilde{C}(\omega)}\,,
\end{equation}
where the last inequality follows from the definition of $\tilde{C}(\omega)$, therefore $C(\omega)\le \tilde{C}(\omega)$, and the claim follows.
\end{proof}
\end{prop}

In the tracial setting \cite{Junge2014}, and more generally whenever $[K,\omega]=0$ the quantity $\tr\exp(K+\ln\omega)$ can be interpreted as the Laplace transform of $K$ in the state $\omega$, and therefore the equivalence between Gaussian concentration and the transportation cost inequality holds. However, this is no longer true when $K$ and $\omega$ do not commute, and the following bound can turn out to be strictly stronger to the transportation cost inequality as a consequence of the Golden-Thompson inequality: for any $K\in\mathcal{O}_V$ such that $\tr[\omega K]=0$,
\begin{align}\label{Laplacebound}
    \ln \tr\left[\omega\,e^K\right]\le  \,\frac{C'(\omega)}{4} \|K\|_L^2\,.
\end{align}
In other words, $C(\omega)'\ge C(\omega)$. Recently, bounds of the form of \eqref{Laplacebound} were obtained for some subclasses of Lipschitz observables $K$ (typically local observables) when the $\omega$ is the Gibbs state of a (possibly non-commuting) quasi-local Hamiltonian $H$ \cite{Kuwahara2020} using cluster expansion techniques. However, the existence of the Gaussian concentration inequality for general Lipschitz observables was left open.

Here instead, we pursue a different approach using our transportation cost inequality. In particular, we prove that \eqref{Laplacebound} can be approximately recovered for Gibbs states of commuting Hamiltonians for a larger class of Lipschitz observables than those considered in \cite{Kuwahara2020}. For this, we adapt the result of \cite[Theorem 8]{Rouz2019} which was written for $W_{1,\nabla}$ to the case of $W_1$. In this section, we denote by $X_R$, respectively $X_I$ the real, respectively imaginary parts of an operator $X\in\mathcal{B}_V$. Given an observable $O\in \mathcal{O}_V$ with spectral decomposition $O:=\sum_{\lambda}\lambda P_\lambda$, a state $\omega\in\mathcal{S}_V$ and a real number $r\in\mathbb{R}$, we denote by
\begin{align}
    \mathbb{P}_\omega(O\ge r):=\sum_{\lambda\ge r}\tr(\omega P_\lambda)
\end{align}
the probability of getting an eigenvalue $\lambda\ge r$ when measuring $O$ on the state $\omega$.

\begin{thm}\label{thmGaussiantail}
Assume that the full-rank state $\omega\in\mathcal{S}_V$ satisfies $\operatorname{TC}(c)$ for some $c>0$. Then, for any observable $O\in\mathcal{O}_V$,
\begin{align}
    \mathbb{P}_{\omega}\big(|O- \tr({\omega O})\,\mathbb{I}|\ge r \big)\le 2\,\exp\left( -\frac{r^2}{4c\max\big\{\|(\omega^{-\frac{1}{2}}O\omega^{\frac{1}{2}})_R\|_{L}^2,\|(\omega^{-\frac{1}{2}}O\omega^{\frac{1}{2}})_I\|_{L}^2\big\}}\right)\,.
\end{align}
Whenever $[O,\omega]=0$ the bound can be tightened into
\begin{align}\label{commutingcaseconcentration}
        \mathbb{P}_{\omega}\big(|O- \tr({\omega O})\,\mathbb{I}|\ge r \big)\le 2\,\exp\left( -\frac{r^2}{c \|O\|_{L}^2}\right)\,.
\end{align}
Therefore, whenever $\omega\in\mathcal{S}_V$ corresponds to the Gibbs state of a local commuting Hamiltonian on a hypergraph at inverse temperature $\beta$, the above bounds hold as long as $0<\beta< \beta_c$ where $\beta_c$ is defined in \eqref{boundbetac}. 
\end{thm}

\begin{proof}
Given $X\in\mathcal{B}_V$, we denote by $X:=X_R+iX_I$ its decomposition onto real and imaginary parts. We also assume that $\tr(\omega X)=0$ and $\|X_R\|_L,\|X_I\|_L\le 1$. By assumption, we have that for any $\rho\in\mathcal{S}_V$
\begin{align}
    \left|\tr(\rho X)\right|\le \left|\tr(\rho X_R)\right|+\left|\tr(\rho X_I)\right|\le 2\left\|\rho-\omega\right\|_{W_1}\le 2\sqrt{c\,S(\rho\|\omega)}\,.
\end{align}
 Then, since $\inf_{\theta>0}\Big(\frac{a}{\theta}+\frac{b\theta}{2}\Big)=\sqrt{2ab}$ for any $a,b\ge 0$, we have that for all $\theta>0$:
 \begin{align}
     \big|\tr(\rho X)\big|\le \sqrt{2}\,\frac{S(\rho\|\omega)}{\theta}+\frac{c\,\theta}{\sqrt{2}}\qquad\Leftrightarrow\qquad \theta\big|\tr(\rho X)\big|- \frac{c}{\sqrt{2}}\,\theta^2\le \sqrt{2}\,S(\rho\|\omega)\,.
 \end{align}
Next, we further upper bound the relative entropy in terms of the maximal divergence $\widehat{S}(\rho\|\omega):=\tr\big[\omega\big(\omega^{-\frac{1}{2}}\rho\omega^{-\frac{1}{2}}\big)\ln(\omega^{-\frac{1}{2}}\rho\omega^{-\frac{1}{2}}) \big]$ \cite{matsumoto2015new}. Choosing $\rho=\omega^{\frac{1}{2}}e^{\theta O}\omega^{\frac{1}{2}}/\tr(\omega e^{\theta O})$ for some observable $O\in\mathcal{O}_V$, we arrive at
\begin{align}
    \theta\big|\tr(\rho X)\big|- \frac{c}{\sqrt{2}}\,\theta^2\le \sqrt{2}\,\theta \frac{\tr(\omega e^{\theta O}O)}{\tr(\omega e^{\theta O})}-\sqrt{2}\ln\big(\tr(\omega e^{\theta O})\big)\,.
\end{align}
Next, we choose $X=\sqrt{2}\omega^{-\frac{1}{2}}O\omega^{\frac{1}{2}}$, so that the previous inequality reduces to 
\begin{align}
\ln\big(\tr(\omega e^{\theta O})\big)\le \frac{c}{2}\,\theta^2\,.
\end{align}
The above inequality can be interpreted as a bound on the log-Laplace transform of the non-commutative variable $O$ in the state $\omega$. By a use of Markov's inequality followed by an optimization over the variable $\theta>0$, we finally get
\begin{align}
    \mathbb{P}_\omega\big(\big|O\big|\ge r\big)\le 2e^{-\frac{r^2}{2c}}\,.
\end{align}
The result follows after simple rescalings. The tightening in the case of an observable commuting with $\omega$ can be found by following the same steps as the ones above. 
\end{proof}

In general, there is no way to precisely relate the Lipschitz constants of the real and imaginary parts of  $\omega^{-\frac{1}{2}}O\omega^{\frac{1}{2}}$ to the Lipschitz constant of $O$ when $[O,\omega]\ne 0$. In the next result, we however prove that the constants have similar scalings in the case of a commuting Gibbs measure $\omega$ of a local Hamiltonian. 

\begin{lem}
Let $O=\sum_{A\subseteq V} \lambda_A\,O_A\otimes \mathbb{I}_{A^c}$ be the decomposition of an observable $O$ in $\mathcal{O}_V$, where for each subregion $A$, $O_A$ is exactly supported in $A$ with $\|O_A\|_\infty\le 1$, and $\lambda_A\in\mathbb{R}$. Let further $\omega$ be the Gibbs state of a geometrically $k$-local, commuting Hamiltonian $H_V:=\sum_{B\subset V}h_B\otimes\mathbb{I}_{B^c}$ at inverse temperature $\beta$. Then,
\begin{align}
    \|(\omega^{-\frac{1}{2}}O\omega^{\frac{1}{2}})_{R}\|_{L}\,,~~   \|(\omega^{-\frac{1}{2}}O\omega^{\frac{1}{2}})_{I}\|_{L}\le 4\max_{i\in V}\,\sum_{\substack{A\subset V\\A_{\partial k}\ni i}}\,|\lambda_A|\,\exp\Big(\beta\sum_{B\cap A\ne \emptyset}\|h_B\|_\infty\Big)\,,
\end{align}
where $A_{\partial k}:=\{j\in V:\, \operatorname{dist}(j,A)\le k\}$ denotes the $k$-enlargement of $A$. In particular, whenever the state $\omega$ satisfies $\operatorname{TC}(c)$ with $c=O(n)$, any local observable $O$ gives rise to a sub-Gaussian random variable variance $O(n)$ when measured in the state $\omega$.
\end{lem}
\begin{proof}
We prove the bound for the real part of $(\omega^{-\frac{1}{2}}O\omega^{\frac{1}{2}})_{R}$ since the proof for the imaginary part follows the exact same reasoning. First, by \autoref{prop15}, since for any $A\subset V$, $(\sigma^{-\frac{1}{2}}O_A\sigma^{\frac{1}{2}})_R$ is supported in region $A_{\partial k}$, we have that
\begin{align}
     \|(\omega^{-\frac{1}{2}}O\omega^{\frac{1}{2}})_{R}\|_{L}&\le 2\,\max_{i\in V}\,\Big\|\sum_{\substack{A\subset V\\ A_{\partial k}\ni i}}  \lambda_A\,\Big[(\omega^{-\frac{1}{2}}O_A\omega^{\frac{1}{2}})_R-\tr_i\big[(\omega^{-\frac{1}{2}}O_A\omega^{\frac{1}{2}})_R \big]\otimes \frac{\mathbb{I}_d}{d} \Big]\Big\|_\infty\\
     &\le 4\,\max_{i\in V}\sum_{\substack{A\subset V\\A_{\partial k\ni i}}}\,|\lambda_A|\,\|\omega^{-\frac{1}{2}}O_A\omega^{\frac{1}{2}}\|_\infty\\
     &\le 4\,\max_{i\in V}\sum_{\substack{A\subset V\\A_{\partial k\ni i}}}\,|\lambda_A|\,\|e^{\beta\sum_{B\cap A\ne \emptyset}h_B}O_Ae^{-\beta\sum_{B\cap A\ne \emptyset}h_B}\|_\infty\\
  &   \le 4\,\max_{i\in V}\sum_{\substack{A\subset V\\A_{\partial k\ni i}}}\,|\lambda_A|\,e^{\beta\sum_{B\cap A\ne \emptyset}\|h_B\|_\infty}\,.
\end{align}
The result follows. 
\end{proof}

\subsection{Comparison to previous tail bounds}

Our main result can be compared to other recently derived concentration bounds for quantum Gibbs states: in \cite[Corollary 5.4]{Anshu2016}, the authors consider a product state $\rho=\bigotimes_{v\in V}\rho_v\in \mathcal{S}_V$ as well as a Hamiltonian $H=\sum_{A\in E_{k,m}}h_A$, where the set $E_{k,m}$ of subsets of $V$ has the following properties: for any $A\in E_{k,m}$, 
\begin{itemize}
\item[(i)] $|A|\le k$; 
\item[(ii)] $\big|\{A'\in E_{k,m}:\,A\cap A'\ne\emptyset\}\big|\le m$.
\end{itemize}
With these conditions, he was able to prove that 
\begin{align*}
    \mathbb{P}_{\rho}\big(|H-\tr(\rho H)|\ge r\big)\le 2e^{-\frac{r^2}{4eN^3kn}}\,,
\end{align*}
where number $N:=\max_{v\in V}\sum_{A\in E_{k,m}|v\in A}\,1$ is the number of local terms acting non-trivially on spin $v$. A similar bound was previously derived by Kuwahara \cite[Theorem 7]{Kuwahara2016}, under a notion of $g$-extensivity: a local Hamiltonian $H$ is said to be $g$-extensive if for every spin $v$, $\sum_{A\in E_{k,m}|\,v\in A} \|h_A\|_\infty\le g$. Under this condition, he shows that 
\begin{align*}
    \mathbb{P}_\rho\big( |H-\tr(\rho H)|\ge r\big)=\mathcal{O}(1)\,e^{-\frac{r^2}{cn\log(\frac{r}{\sqrt{n}})}}\,,
\end{align*}
where $c$ is a $\mathcal{O}(1)$ constant which depends only on $k$ and $g$. Although these results recover the Gaussian tails of our \autoref{thmGaussiantail} (up to logarithmic overheads), they only work for tensor product states and a subclass of Lipschitz observables. In particular, the tails become trivial whenever the Hamiltonian is a sum of terms acting on non-intersecting regions $A$ of arbitrary size. In contrast, our bound is still non-trivial for this class of observables, since their Lipschitz constant is still $\mathcal{O}(1)$. 

More recently, Kuwahara and Saito derived new concentration bounds for Gibbs states of interacting Hamiltonians in order to study the problem of equivalence of quantum statistical ensembles \cite{kuwahara2020eigenstate,Kuwahara2020} (see \autoref{ETH}): in \cite{kuwahara2020eigenstate} first, the authors consider a Gibbs state $\omega$ of a local Hamiltonian on a $D$-dimensional regular lattice $\mathbb{Z}^D$. They further assume the following $(r_0,\xi)$ clustering: for any operators $O_A,O_B$ supported on the subsets $A$ and $B$,
\begin{align}\label{clustering}
    |\tr(\omega O_AO_B)-\tr(\omega O_A)\tr(\omega O_B)|\le \|O_A\|_\infty\,\|O_B\|_\infty\,e^{-\operatorname{dist}(A,B)/\xi},
\end{align}
whenever $\operatorname{dist}(A,B)\ge r_0$. Under this condition, they were able to show in Equation (S.17) (see also \cite[Theorem 4.2]{Anshu2016} for a similar bound)
\begin{align*}
    \mathbb{P}_\omega\big(|O-\tr(\omega O)|\ge r\big)\le \min\big\{1,(e+3e\xi)\max\big(e^{-(r^2/(cn))^{\frac{1}{D+1}}},\,e^{-\frac{r^2}{c' \ell^Dn}}\big)\big\}\,,
\end{align*}
for some constants $c,c'$ further depending on $\xi$ and $D$, and where $\ell$ denotes the locality of the observable $O$. Therefore, and although the clustering of correlations is known to hold at high enough temperature \cite{Kliesch14}, the bound is suboptimal for two reasons: firstly, whenever $r$ is small enough, the exponent has the worse scaling $r^{2/(D+1)}$. Secondly, the bounds baddly dependence on the locality of $O$, and becomes trivial whenever $O$ is a sum of highly non-local terms. This second limitation also holds for the Gaussian concentration bound found in \cite[Corollary 1]{Kuwahara2020} for 
high-temperature Gibbs states of Hamiltonians with long-range interactions. In comparison to the works cited above, our bound always provides better dependence of the tail on the locality of the observable, albeit under the condition that the Hamiltonian is made of local commuting terms.

\section{Equivalence of statistical mechanical ensembles}\label{ETH}
The three main ensembles employed in quantum statistical mechanics to compute the equilibrium properties of quantum systems are the canonical ensemble, the microcanonical ensemble and the diagonal ensemble.
The quantum state associated to the canonical ensemble is the Gibbs state, which describes the physics of a system that is at thermal equilibrium with a large bath at a given temperature.
The diagonal and microcanonical ensembles both describe the physics of an isolated quantum system, and the associated states are convex combinations of the eigenstates of the Hamiltonian.
The microcanonical ensemble assumes a uniform probability distribution for the energy in a given energy shell.
The diagonal ensemble includes all the states that are diagonal in the eigenbasis of the Hamiltonian, and in particular it includes the eigenstates themselves.

For many quantum systems, the microcanonical and canonical ensembles give the same expectation values for local observables if the corresponding states have approximately the same average energy.
A lot of effort has been devoted to determining conditions under which the two ensembles are equivalent \cite{lima1971equivalence,lima1972equivalence,muller2015thermalization,brandao_equivalence_2015}.
The most prominent among such conditions are short ranged interactions and a finite correlation length, but analytical proofs can be obtained only in the case of regular lattices \cite{brandao_equivalence_2015}.
The situation is more complex for the diagonal ensemble.
The condition under which this ensemble is equivalent to the microcanonical and canonical ensembles is called Eigenstate Thermalization Hypothesis (ETH) \cite{deutsch1991quantum,srednicki_chaos_1994,gogolin2016equilibration,reimann2018dynamical,de2015necessity}, and states that the expectation values of local observables on the eigenstates of the Hamiltonian are a smooth function of the energy, \emph{i.e.}, for any given local observable, \emph{any} two eigenstates with approximately the same energy yield approximately the same expectation value.
The ETH is an extremely strong condition on the Hamiltonian and several quantum systems, including all integrable systems, do not satisfy it.
A weak version of the ETH has been formulated \cite{biroli2010effect,kuwahara2020eigenstate}, stating that for any given local observable, \emph{most} eigenstates in an energy shell yield approximately the same expectation value, or, more precisely, that the fraction of eigenstates yielding expectation values far from the Gibbs state with the same average energy vanishes in the thermodynamical limit.
The weak ETH implies the equivalence between the canonical and microcanonical ensembles, but is not sufficient to prove their equivalence with the diagonal ensemble.
Under the hypothesis of finite correlation length in the Gibbs state, an analytical proof of the weak ETH is available only for regular lattices \cite{kuwahara2020eigenstate}.

A connection between a transportation cost inequality and the ETH was made by one of the authors in the case of a regular lattice and a nearest neighbour Hamiltonian \cite{capel2020modified}.
Here we look at the general problem of the equivalence of the statistical mechanical ensembles and of the weak ETH from the perspective of optimal mass transport, and show that such equivalence can be formulated as closeness of the respective states in the $W_1$ distance.
The closeness in the $W_1$ distance implies closeness of the expectation values of all Lipschitz observables, which constitute a significantly larger class than local observables.
Therefore, the perspective of optimal mass transport can significantly extend the previous results.
Moreover, we will show that the equivalence of the ensembles is intimately linked to the constant of the transportation cost inequality for the Gibbs states.

As in the rest of the paper, we consider a quantum system made by $n$ qudits located at the vertices of a graph with vertex set $V$.
Let us assume that a Gibbs state $\omega\in\mathcal{S}_V$ satisfies the transportation cost inequality with a constant
\begin{equation}
    C(\omega) \le n\,C\,,
\end{equation}
where $C$ does not depend on $n$.
This condition is satisfied under the hypotheses of \autoref{thm:Markov}, \autoref{thmTCcurvature} or \autoref{thm:LS}.
We stress that, contrarily to the results of Refs. \cite{brandao_equivalence_2015,kuwahara2020eigenstate}, the condition does not require us to restrict to regular lattices, since \autoref{thmTCcurvature} does not need this hypothesis.
The following \autoref{prop:equiv} implies that any state $\rho\in\mathcal{S}_V$ is close in $W_1$ distance to the Gibbs state $\omega$ with the same average energy, provided that $\rho$ and $\omega$ have approximately the same entropy, \emph{i.e.},
\begin{equation}
    S(\omega) - S(\rho) \ll n\,.
\end{equation}
Moreover, under the same hypothesis, the average reduced states over one qudit of $\rho$ and $\omega$ are close in trace distance.
\begin{prop}\label{prop:equiv}
Let $\omega\in\mathcal{S}_V$ be a Gibbs state for the Hamiltonian $H\in\mathcal{O}_V$.
Then, any quantum state $\rho\in\mathcal{S}_V$ with the same average energy as $\omega$ satisfies
\begin{equation}
    \frac{1}{n}\left\|\rho-\omega\right\|_{W_1} \le \sqrt{C\,\frac{S(\omega) - S(\rho)}{n}}\,.
\end{equation}
Moreover, let $\Lambda:\mathcal{O}_V\to\mathcal{O}(\mathbb{C}^d)$ be the quantum channel that computes the average marginal state over one qudit, \emph{i.e.}, for any $\rho\in\mathcal{S}_V$,
\begin{equation}
    \Lambda(\rho) = \frac{1}{n}\sum_{v\in V}\rho_v\,.
\end{equation}
Then,
\begin{equation}
    \left\|\Lambda(\rho) - \Lambda(\omega)\right\|_1 \le 2\sqrt{C\,\frac{S(\omega) - S(\rho)}{n}}\,.
\end{equation}
\end{prop}
\begin{proof}
We have from the transportation cost inequality
\begin{equation}
    \left\|\rho-\omega\right\|_{W_1} \le \sqrt{C(\omega)\,S(\rho\|\omega)} = \sqrt{C(\omega)\left(S(\omega) - S(\rho)\right)}\,,
\end{equation}
where the last equality follows since $\mathrm{Tr}\left[\rho\ln\omega\right] = \mathrm{Tr}\left[\omega\ln\omega\right]$.

We have from \autoref{cor1}
\begin{align}
    \left\|\Lambda(\rho) - \Lambda(\omega)\right\|_1 \le \frac{1}{n}\sum_{v\in V}\left\|\rho_v - \omega_v\right\|_1 \le \frac{2}{n}\left\|\rho - \omega\right\|_{W_1} \le \frac{2}{n}\sqrt{C(\omega)\left(S(\omega) - S(\rho)\right)}\,,
\end{align}
and the claim follows.
\end{proof}

Choosing $\rho$ to be diagonal in the eigenbasis of the Hamiltonian, \autoref{prop:equiv} implies that any convex combination of a sufficiently large number of eigenstates is close in $W_1$ distance to the Gibbs state with the same average energy.
Such number of eigenstates can even be an exponentially small fraction of the total number of eigenstates appearing in a microcanonical state, since the uniform superposition of a fraction $e^{-n\epsilon}$ of the eigenstates decreases the entropy by $\epsilon\,n$.
Therefore, \autoref{prop:equiv} constitutes an exponential improvement over the weak ETH.

A natural question is whether also the strong ETH can be captured by the $W_1$ distance.
Unfortunately the answer is negative.
Indeed, proving the strong ETH via optimal mass transport would mean to prove that all the eigenstates of the Hamiltonian are close in $W_1$ distance to the Gibbs states with the corresponding average energy.
However, \autoref{thm1} implies that any state with low entropy, and in particular any pure state, is far from any state with large entropy, and in particular from a Gibbs state with temperature $\Omega(1)$.
More precisely, for any two states $\rho,\,\omega\in\mathcal{S}_V$,
\begin{equation}\label{eq:S}
    \left\|\rho-\omega\right\|_{W_1} \ge \frac{S(\omega) - S(\rho) - \ln\left(n+1\right) - 1}{\ln\left(d^2n\right)}\,.
\end{equation}
Equation \eqref{eq:S} also implies that any quantum state which is close in $W_1$ distance to the Gibbs state with the same average energy must have approximately also the same entropy, and in this sense \autoref{prop:equiv} is optimal.

\subsection{Comparison with previous results}
To make our result more easily comparable to the literature, let us introduce more formally the microcanonical ensemble: given the decomposition $H=\sum_E EP(E)$, we define the microcanonical ensemble state 
\begin{align*}
    \omega_{E,\Delta}:=\frac{P(E,\Delta)}{\tr(P(E,\Delta))}\,,
\end{align*}
where $P(E,\Delta)$ corresponds to the projection onto the subspace spanned by the eigenvectors whose energy belongs to the interval $(E-\Delta,E]$.

\begin{cor}\label{coro:equivalence}
Assume the Gibbs state $\omega$ satisfies $C(\omega) \le Cn$. Then for any Lipschitz observable $O$,
\begin{align*}
\frac{1}{n}\,\big|\tr(\omega O)-\tr(\omega_{E,\Delta} O)\big|\le \|O\|_L\,o_{n\to\infty}(1)\,.
\end{align*}

\end{cor}

\begin{proof}
In view of \autoref{prop:equiv}, it suffices to control the relative entropy between the microcanonical and canonical ensemble states. Then, 
\begin{align}\label{eq:relativeentropy}
    S(\omega_{E,\Delta}\|\omega)=\ln\frac{\tr(e^{-\beta H})}{\tr(P(E,\Delta))}+\beta\tr\Big[H\,\frac{P(E,\Delta)}{\tr(P(E,\Delta))}\Big]&\le \beta E+\ln\Big[\frac{\tr(e^{-\beta H})}{\tr(P(E,\Delta))}\Big]\,.
    \end{align}
    Next, we control the ration $\frac{\tr(e^{-\beta H})}{\tr(P(E,\Delta))}$. For this, we use an argument which was already used in \cite[Equation (S.56)]{kuwahara2020eigenstate}:
   First, we have found in \eqref{commutingcaseconcentration} that 
   \begin{align*}
       \mathbb{P}_\omega(|H-\tr(\omega H)\mathbb{I}|\ge r)\le 2\,e^{-\frac{r^2}{Cn\|H\|_L^2}}\,.
   \end{align*}
   Therefore, choosing the interval $\tilde{\Delta}:=(\tr(\omega H)-\sqrt{Cn\ln(4)}\|H\|_L,\tr(\omega H)+\sqrt{Cn\ln(4)}\|H\|_L]$, we have
   \begin{align*}
       \tr\Big[\sum_{E\in \tilde{\Delta}}\,\frac{e^{-\beta E}}{\tr(e^{-\beta H})} P(E)\Big]\ge \frac{1}{2} 
   \end{align*}
  Next, we define 
\begin{align} \label{eq:ZZtilde}
\tilde{Z}:=\tr\Big[\sum_{E\in\tilde{\Delta}} e^{-\beta E}P(E)\Big]\ge \frac{\tr(e^{-\beta H})}{2}
\end{align}
    Choosing a slightly extended interval $\tilde{\Delta}':=(\tr(\omega H)-\sqrt{Cn\ln(4)}\|H\|_L-\Delta,\tr(\omega H)+\sqrt{Cn\ln(4)}\|H\|_L+\Delta]$, we have 
    \begin{align}\label{eq:Ztildebound}
        \tilde{Z}\le \sum_{\nu\in\mathbb{Z}:\,\nu\Delta\in \tilde{\Delta}'}\tr(P(\nu\Delta,\Delta))\,e^{-\beta\Delta(\nu-1)}
    \end{align}
    Now, for $E:=\operatorname{argmax} e^{-\beta E}\tr(P(E,\Delta))$, we have 
    \begin{align*}
        \tr(\nu\Delta,\Delta)e^{-\beta \delta(\nu-1)}\le e^{\beta\Delta}e^{-\beta E}\,\tr(P(E,\Delta))\,.
    \end{align*}
    Replacing in \eqref{eq:Ztildebound}, we have that 
    \begin{align*}
        \tilde{Z}\le e^{\beta\Delta}(2+2\sqrt{Cn\ln(4)}\|H\|_L/\Delta)e^{-\beta E}\tr(P(E,\Delta))
    \end{align*}
    Finally, using the lower bound \eqref{eq:ZZtilde}, we have that
    \begin{align*}
        \ln\Big[\frac{\tr(e^{-\beta H})}{\tr(P(E,\Delta))}\Big]\le \beta(\Delta-E)+\ln(4+4\sqrt{Cn\ln(4)}\|H\|_L/\Delta)
    \end{align*}
    Therefore, plugging this last bound into \eqref{eq:relativeentropy}, we have found that
    \begin{align*}
        S(\omega_{E,\Delta}\|\omega)\le \beta\Delta+\ln(4+4\sqrt{Cn\ln(4)}\|H\|_L/\Delta)\,.
    \end{align*}
    Therefore, $ S(\omega_{E,\Delta}\|\omega)=o(n)$ whenever $\Delta=e^{-o(n)}$, and the result follows.
    
    \end{proof}
In \cite[Theorem 2]{kuwahara2020eigenstate}, it is showed that, under the $(r_0,\xi)$-clustering of correlations \eqref{clustering}, for any observable $O:=\sum_{v}O_v$ where each $O_v$ acts on spin $v$ as well as other spins $w$ with $\operatorname{dist}(v,w)\le \ell$ and has $\|O_v\|_\infty\le 1$, 
\begin{align*}
  \frac{1}{n}\,  \big|\tr(\omega_{E,\Delta} O)-\tr(\omega O) \big|\le \frac{1}{\sqrt{n}}\, \max(c_1B_1,c_2B_2),
\end{align*}
where $B_1:=\log(\sqrt{n}/\Delta)^{\frac{d+1}{2}}$, $B_2:=(\ell^D\log(\sqrt{n}/\Delta))^{\frac{1}{2}}$, and the constants $c_1$ and $c_2$ depend on $D,r_0,\xi$ and the locality $k$ of $H$. Therefore, as long as the energy shell $\Delta$ is chosen as $\Delta\sim e^{-\mathcal{O}(n^{\frac{1}{D+1}})}$ the averages of the operator density $\frac{O}{n}$ in the canonical and microcanonical ensemble states converge to the same number as $n\to\infty$. Similar bounds were also derived in \cite[Corollary 3]{Kuwahara2020} for larger classes of non-local Hamiltonians and observables above some threshold temperature. \autoref{coro:equivalence} constitutes an improvement over these results in two senses: Firstly, it applies to a more general class of Lipschitz observables. Secondly, it allows for a smaller energy shell $\Delta=e^{-o(n)}$. However, the condition $\operatorname{TC}(Cn)$ is currently only known to hold for the smaller class of local commuting Hamiltonians.

\section*{Acknowledgements}

The research of CR has been supported by project QTraj (ANR-20-CE40-0024-01) of the French National Research Agency (ANR) and by  a  Junior  Researcher  START  Fellowship  from  the MCQST. 

\bibliographystyle{unsrt}
\bibliography{markov}

\end{document}